\documentclass[11pt,a4paper]{article}

\usepackage[margin=2.0cm]{geometry}
\usepackage{amsmath}
\usepackage{amssymb}
\usepackage{amsthm}
\usepackage{dsfont}
\usepackage{color}
\usepackage{amsmath}
\usepackage{amsfonts}
\usepackage{amssymb}
\usepackage{bbm}
\usepackage{hhline}
\usepackage{tikz}
\usepackage{pgfplots}
\usepackage{subcaption}

\usetikzlibrary{
	positioning,
	plotmarks,
	shapes.callouts,	
}

\tikzset{default node/.style={
	draw, 
	circle,
	inner sep=0mm,
	minimum size=5mm,
	very thick,
	font=\small,
	black!70,
}}

\usepackage[procnumbered, linesnumbered,algoruled,noend ]{algorithm2e}

\newcommand{\cout}[1]{}
\newcommand{\bbeta}{{\bar{\beta}}}
\newcommand{\floor}[1]{{\left\lfloor {#1} \right\rfloor}}
\newcommand{\ceil}[1]{{\left\lceil {#1} \right\rceil }}
\newcommand{\remove}[1]{}

\newcommand{\cA}{{\mathcal{A}}}

\newcommand{\supp}{{\textnormal{supp}}}

\newcommand{\cC}{{\mathcal{C}}}
\newcommand{\bx}{{\bar{x}}}

\newcommand{\bgam}{{\bar{\gamma}}}
\newcommand{\blam}{{\bar{\lambda}}}
\newcommand{\bnu}{{\bar{\nu}}}
\newcommand{\by}{{\bar{y}}}
\newcommand{\bz}{{\bar{z}}}
\newcommand{\bc}{{\bar{c}}}

\newcommand{\ba}{{\bar{a}}}
\newcommand{\OPT}{\textnormal{OPT}}
\newcommand{\ett}{\eta(T)}

\newcommand{\eps}{{\varepsilon}}
\newcommand{\aeps}{{\mu}}
\newcommand{\E}{{\mathbb{E}}}

\newcommand{\ariel}[1]{{\color{red} (Ariel :#1)}}

\def \II   {{\mathcal I}}
\newcommand{\one}{\mathbbm	{1}}

\newcommand{\type}{\textnormal{type}}
\newcommand{\pipage}{\textsf{Pipage}}

\newcommand{\cT}{\mathcal{T}}

\begin{document}

\newtheorem{thm}{Theorem}[section]
\newtheorem{prop}[thm]{Proposition}
\newtheorem{assm}[thm]{Assumption}
\newtheorem{lem}[thm]{Lemma}

\newtheorem{obs}[thm]{Observation}
\newtheorem{cor}[thm]{Corollary}
\newtheorem{lemma}[thm]{Lemma}
\newtheorem{proposition}[thm]{Proposition}
\newtheorem{claim}[thm]{Claim}
\newtheorem{defn}[thm]{Definition}
\newtheorem{definition}[thm]{Definition}

	\title{
		Modular and Submodular Optimization with 
		Multiple Knapsack Constraints via Fractional Grouping}
	\author{
		Yaron Fairstein\thanks{Computer Science Department, Technion, Haifa 3200003,
			Israel. \mbox{E-mail: {\tt yyfairstein@gmail.com}}}
		\and
		Ariel Kulik\thanks{Computer Science Department, Technion, Haifa 3200003,
			Israel. \mbox{E-mail: {\tt kulik@cs.technion.ac.il}}}
		\and
		Hadas Shachnai\thanks{Computer Science Department, Technion, Haifa 3200003,
			Israel. \mbox{E-mail: {\tt hadas@cs.technion.ac.il}}.
		}
	}
	\date{}
	\maketitle
\begin{abstract}
	
\cout{
A multiple knapsack constraint over a set of items is defined by a set of bins of arbitrary capacities, and a weight for each of the items. An assignment for the constraint is an allocation of subsets of items to the bins which 
adheres to bin capacities.
While problems involving a single knapsack constraint are well understood, relatively
little is known about their natural generalizations 
to problems 
involving 
multiple knapsack constraints.
A classic example is Multiple Choice Knapsack, that is 
known to admit a {\em fully polynomial time approximation scheme}, whereas the best known approximation for 
Multiple Choice Multiple Knapsack  is $2$.  Another example is Non-monotone Submodular Knapsack  for which a $(0.385-\eps)$-approximation is known, 
with no known approximation guarantee  for Non-monotone Submodular Multiple Knapsack.

In this paper we mitigate this gap. Our results include, among others, a {\em polynomial-time approximation scheme (PTAS)} for Multiple Choice Multiple Knapsack,  a $(0.385-\eps)$-approximation for Non-monotone Submodular Multiple Knapsack, 
 a $(1-e^{-1}-o(1))$-approximation for  Monotone Submodular Multiple Knapsack  with Uniform Capacities, and a PTAS
for modular optimization with two multiple knapsack constraints. We obtain our results using a novel {\em fractional grouping} technique, 
which partitions items into groups based on a fractional solution for the given instance. Subsequently, items of the same group can be treated as if they have identical weights. 
}

A multiple knapsack constraint over a set of items is defined by a set of bins of arbitrary capacities, and a weight for each of the items. An assignment for the constraint is an allocation of subsets of items to the bins which 
adheres to bin capacities. In this paper we present a unified algorithm
that yields efficient approximations for a wide class of submodular and modular optimization problems involving multiple knapsack constraints.
One notable example is a {\em polynomial time approximation scheme} for Multiple-Choice Multiple Knapsack, improving upon the best known ratio of $2$.
Another example is Non-monotone Submodular Multiple Knapsack, for which we 
obtain a $(0.385-\eps)$-approximation, matching the best known ratio for a single knapsack constraint.  The robustness of our algorithm is achieved by applying
a novel {\em fractional} variant of the classical linear grouping technique, which is of independent interest.

\end{abstract}


\section{Introduction}
\label{sec:intro}

The Knapsack problem is one of the most studied problems in mathematical programming and 
combinatorial optimization, with applications ranging from power management 
and production planning, to blockchain storage allocation and key generation in cryptosystems~\cite{MT90,KPP04,van2010performance,xu2018blockchain}.
In a more general form, knapsack problems require assigning items of various sizes (weights) to a set of bins (knapsacks) of bounded capacities. The bin capacities then constitute
the hard constraint for the problem. 
Formally, a {\em multiple knapsack constraint} (MKC) over a set of items is defined by a collection of bins of varying capacities and a non-negative weight for each item.
A feasible solution for the constraint is an assignment of subsets of items to the bins, 
such that the total weight of items assigned to each bin does not exceed its capacity.
This constraint plays a central role in the classic 
Multiple Knapsack problem \cite{CK05,Ja10, Ja12}.
The input is an MKC and each item also has a profit. The objective is to find a feasible solution for the MKC such that the total profit of assigned items is maximized.

Multiple Knapsack can be viewed as a maximization variant of the Bin Packing problem~\cite{KK82,FL81}. In Bin Packing we are given a set of items,
each associated with non-negative weight. We need to pack the items into a minimum number of identical (unit-size) bins.

A prominent technique for approximating Bin Packing is {\em grouping}, which decreases the number of distinct weights in the input instance.
Informally, a subset of items is partitioned into groups $G_1, \ldots, G_{\tau}$, and all the items within a group are treated as if they have the 
{\em same} weight (e.g., \cite{FL81,KK82}).
By properly forming the groups, the increase in the 
number of bins required for packing the instance
 can be bounded. Classic grouping techniques require knowledge of the items to be packed, and thus cannot be easily applied in the context of  maximization problems, and specifically for a multiple knapsack constraint.

\cout{
While grouping is a powerful technique in the context of Bin Packing, it cannot be trivially applied to problems involving a multiple knapsack constraint, as it assumes knowledge of 
the set of items to be packed $-$ an assumption which
does not hold for
 maximization problems.}

The main technical contribution of this paper is the introduction of {\em fractional grouping}, a variant of linear grouping which can be applied to multiple knapsack constraints. 
Fractional Grouping partitions the items into groups using  an easy to obtain fractional solution, bypassing the requirement to know the items in the solution.

Fractional Grouping proved to be a  robust technique for maximization problems. We use the technique to obtain, among others, 
 a {\em polynomial-time approximation scheme (PTAS)} for the Multiple-Choice Multiple Knapsack Problem,  a $(0.385-\eps)$-approximation for non-monotone submodular maximization with a multiple knapsack constraint, and a $\left(1-e^{-1}-o\left( 1 \right) \right)$-approximation for the Monotone Submodular Multiple Knapsack Problem with Uniform Capacities.
 


\subsection{Problem Definition}
\label{sec:prob_def}
We first define formally key components of the problem studied in this paper. 

A {\em multiple knapsack constraint} (MKC) over a set $I$ of items, denoted by $\mathcal{K}=(w,B,W)$, is defined by a weight function $w:I\rightarrow \mathbb{R}_{\geq 0}$, a set of bins $B$ and bin capacities given by $W:B\rightarrow \mathbb{R}_{\geq 0}$. An {\em assignment} for the constraint is a function $A:B\rightarrow 2^I$ which assigns a subset of items to each bin. An assignment $A$ is {\em feasible} if $\sum_{i\in A(b)} w(i) \leq W(b)$ for all $b\in B$. We say that $A$ is an {\em assignment of $S$} if $S=\bigcup_{b\in B} A(b)$. 


A set function $f:2^I\rightarrow \mathbb{R}$ is {\em submodular}  if for any $S\subseteq T\subseteq I$ and $i\in I\setminus T$ it holds that  $f(S \cup \{i\}) - f(S) \geq f(T \cup \{i\}) - f(T)$.\footnote{Alternatively, for every $S,T\subseteq I$: $f(S)+f(T)\geq f(S\cup T) + f(S\cap T)$.} 
Submodular functions naturally arise in 
numerous settings.
 While many submodular functions, such as coverage \cite{Fe98} and  matroid rank function \cite{CCPV07}, are monotone, 
i.e., for any $S\subseteq T \subseteq I$, $f(S)\leq f(T)$,
this is not always the case (cut functions \cite{FG95} are a classic example).
 A special case of submodular functions is {\em modular} (or, {\em linear}) functions
  in which, for any $S\subseteq T\subseteq I$ and $i\in I\setminus T$, we have $f(S \cup \{i\}) - f(S) = f(T \cup \{i\}) - f(T)$.
  
\cout{
We consider problems involving submodular maximization with a constant number of MKCs and  possibly an  additional
 constraint $S\in \II$ for some $\II\subseteq 2^I$.
The additional constraint relates to
the properties of the objective function $f$. 
A pair $(\II,f)$ is {\em valid} if $\II\subseteq 2^I$, $f:2^I\rightarrow \mathbb{R}_{\geq 0}$ is submodular and one of the following holds:
\begin{enumerate}
	\item 
	$\II=2^I$ (i.e., $\II$ does not represent a constraint).
	\item The function $f$ is  monotone and  $\II$ defines the independent sets of a matroid.\footnote{A formal definition for matroid can be found in~\cite{Sc03}}
	\item The function $f$ is modular and $\II$ defines
	 the independent sets of a matroid,
	 the intersection of the independent sets of two matroids, or a matching.\footnote{$\II$ is a matching if there is a graph $G=(V,I)$, and $S\in \II$ {\em iff} $S$ is a matching in $G$.}

\end{enumerate}
The definition of a valid pair $(\II,f)$ reflects a technical limitation of our algorithm.
}

The problem of 
{\em Submodular Maximization with $d$-Multiple Knapsack Constraints ($d$-MKCP)} is defined as follows. 
The input is $\mathcal{T} = \left(I, \left( \mathcal{K}_t\right)_{t=1}^d, \II, f \right)$, where $I$ is a set of items,  $\mathcal{K}_t$, $1\leq t\leq d$ are  $d$ MKCs over $I$, $\II\subseteq 2^I$ and $f:2^I\rightarrow \mathbb{R}_{\geq 0}$ is a non-negative submodular function.
$\II$ is an additional constraint which can be one of the following:
$(i)$ 
$\II=2^I$, i.e., any subset of items can be selected. $(ii)$ $\II$ is the independent set of a matroid,\footnote{A formal definition for matroid can be found in~\cite{Sc03}.} or $(iii)$ $\II$ is the intersection of independent sets of two matroids, or $(iv)$ $\II$ is a matching.\footnote{$\II$ is a matching if there is a graph $G=(V,I)$, and $S\in \II$ {\em iff} $S$ is a matching in $G$.} A solution for the instance is $S\in \II$ and $(A_t)_{t=1}^{d}$, where $A_t$ is a feasible assignment of $S$ w.r.t $\mathcal{K}_t$ for $1\leq t\leq d$. The value of the solution is $f(S)$, and the objective is to find a solution of maximal value.

We assume the function $f$ is given via a value oracle.
We further assume that the input indicates the type of constraint that  $\II$ represents. 
Finally, $\II$ is given via a
membership oracle, and if $\II$ is a matroid intersection, a 
membership oracle is given for each matroid.

We refer to the special case in which $f$ is monotone (modular) as monotone (modular) $d$-MKCP. Also, we use non-monotone $d$-MKCP when referring to general  $d$-MKCP instances. Similarly, we refer to the special case in which  $\II$ is an independent set of a  matroid (intersection of independent sets of two matroids or a matching) as $d$-MKCP with a matroid (matroid intersection or matching) constraint. If $\II=2^I$ we refer to the problem as $d$-MKCP with no additional constraint. Thus, for example, in instances of modular $1$-MKCP with a matroid constraint the function $f$ is modular and $\II$ is an independent set of a matroid.

Instances of $d$-MKCP naturally arise in various settings (see a detailed example in  Appendix~\ref{app:applications}). 
\subsection{Our Results}
\label{sec:our_results}
Our main results are summarized in the next theorem (see also Table \ref{tab:results}).
\begin{thm}
\label{thm:main}
For any fixed $d\in\mathbb{N}_+$ and $\eps>0$, there is
\begin{enumerate}
\item A randomized PTAS for modular $d$-MKCP $((1-\eps)$-approximation$)$. 
The same holds for this problem with 
a matroid constraint, matroid intersection constraint, or a matching constraint.
\item A polynomial-time random  $(1-e^{-1}-\eps)$-approximation for monotone $d$-MKCP with a matroid constraint.
\item A polynomial-time random $(0.385-\eps)$-approximation for non-monotone $d$-MKCP with no additional constraint. 
\end{enumerate}
\end{thm}

\begin{table}
	\centering
	\begin{tabular}{|c|c|c|c|}
		\hline 
		Type of Additional  & Modular    & Monotone & Non-Monotone  \\ 
		Constraint
		&  Maximization  &  Submodular Max.  &  Sub. Max  \\ 
		\hhline{|=|=|=|=|}
		No additional constraint
		 &  PTAS & $1-e^{-1}-\eps$  & $0.385-\eps$ \\  
		\hline 
		Matroid constraint
		&  PTAS   &  $1-e^{-1}-\eps$   & $-$   \\
		\hline
		$2$ matroids or a matching
		&PTAS & $-$ &$-$ \\ 
		\hline 
	\end{tabular} 
	\caption{Results of Theorem~\ref{thm:main} for $d$-MKCP}
	\label{tab:results}
\end{table}
All of the results are obtained using a single algorithm (Algorithm \ref{alg:restricted}). 
The general algorithmic result encapsulates several important special cases. 
The Multiple-Choice Multiple Knapsack Problem
is a variant of the Multiple Knapsack Problem in which the items are partitioned into classes $C_1,\ldots, C_k$,  and at most one item can be selected from each class. Formally, Multiple-Choice Multiple Knapsack is the special case of modular $1$-MKCP where $\II$ describes a  partition matroid.\footnote{ That is, 
	$\II=\{ S\subseteq I ~|~\forall 1\leq j \leq k:~|S\cap C_j|\leq 1 \}$ where $C_1,\ldots, C_k$  is a partition of $I$.}
The problem has  natural applications in network optimization~\cite{CG14,TKL06}. The best known approximation ratio for the problem is  $2$ due to~\cite{CG14}. This approximation ratio is improved by  Theorem~\ref{thm:main}, as stated in the following.
\begin{thm}
\label{thm:mcmk}
There is a randomized PTAS for the Multiple-Choice Multiple Knapsack Problem.
\end{thm}

While the Multiple Knapsack Problem and the Monotone Submodular Multiple Knapsack Problem are well understood 
\cite{CK05,Ja10,Ja12, FKNRS20, SZZ20}, no results were previously known for the Non-Monotone Submodular Multiple Knapsack Problem, the special case of non-montone $1$-MKCP with no additional constraint.
A constant approximation ratio for the problem is obtained as a special case of Theorem~\ref{thm:main}.
\begin{thm}
	\label{thm:NMSKP}
For any $\eps>0$ there is a randomized $(0.385-\eps)$-approximation for the Non-Monotone Submodular Multiple Knapsack Problem.
\end{thm}

A PTAS for Multistage Multiple Knapsack, a multistage version of the Multiple Knapsack Problem, can be obtained via a reduction to modular $d$-MKCP with a matroid constraint.\footnote{See, e.g., \cite{BET19} for the Multistage Knapsack model.}
Here, to obtain a $(1-O(\eps))$-approximation for the multistage problem, the reduction solves instances of  modular $\Theta\left(\frac{1}{\eps}\right)$-MKCP with a matroid constraint (see~\cite{FKNR21} for details). Beyond the rich set of applications, our ability to derive such a general result is an  evidence for the robustness of fractional grouping, the main technical contribution of this paper.

Our result for modular $d$-MKCP, for $d\geq 2$, generalizes the PTAS for the classic $d$-dimensional Knapsack problem ($\II=2^I$ and $|B_t|=1$ for any $1\leq t\leq d$).
Furthermore, a PTAS is the best we can expect as there is no {\em efficient PTAS (EPTAS)} already for $d$-dimensional Knapsack, unless
$\textnormal{W}[1]=\textnormal{FPT}$~\cite{KS10}. While 
there is a well-known PTAS for Multiple Knapsack~\cite{CK05}, existing techniques do not readily enable handling
additional constraints, such as a matroid constraint.

The approximation ratio obtained for monotone $d$-MKCP is nearly optimal, as for any $\eps>0$  there is no $(1-e^{-1}+\eps)$-approximation for monotone submodular maximization with a cardinality constraint in the oracle model \cite{NW78}. The approximation ratio is also tight under $P\neq NP$ due to the special case of coverage functions \cite{Fe98}. Previous works \cite{FKNRS20,SZZ20} obtained the same approximation ratio for the Monotone Submodular Multiple Knapsack Problem (i.e, monotone $1$-MKCP).
However, as in the modular case, existing techniques are limited to handling a single MKC (with no other constraints). 

In the non-montone case, the approximation ratio is in fact $(c-\eps)$ for any $\eps>0$, where $c>0.385$ is the ratio derived in~\cite{BF19}. 
This approximation ratio 
matches the current best known ratio for non-monotone submodular maximization with a single knapsack constraint \cite{BF19}.
A $0.491$ hardness of approximation bound for non-monotone $d$-MKCP follows from 
 \cite{GV11}.

The {\em  Monotone Submodular Multiple Knapsack Problem with Uniform Capacities} (USMKP) is the special case of $d$-MKCP
in which $\II=2^I$, $d=1$, $f$ is monotone, and furthermore, all the bins in the MKC have the same capacity. 
That is, $\mathcal{K}_1= (w, B, W)$ and $W(b_1 )= W(b_2)$ for any $b_1, b_2\in B$. 
This restricted variant of $d$-MKCP commonly arises in real-life applications (e.g., in file assignment to several identical storage devices). 
The best known approximation ratio for 
USMKP 
is $(1-e^{-1}-\eps)$ for any fixed $\eps>0$~\cite{FKNRS20,SZZ20}. Another contribution of this paper is an improvement of this ratio. 
\begin{thm}
\label{thm:uniform}
There is a polynomial-time random $\left(1-e^{-1}-O\left( \left(\log |B|\right) ^{-\frac{1}{4}} \right) \right)$-approximation for the Monotone Submodular Multiple Knapsack Problem with Uniform Capacities.
\end{thm}

\subsection{Related Work}
\label{sec:related}

In the classic Multiple Knapsack problem,  the goal is to maximize
a modular set function subject to a single multiple knapsack constraint.
A  PTAS for the problem was first presented by Chekuri and Khanna~\cite{CK05}. The authors also ruled out the existence of a 
{\em fully polynomial time approximation scheme (FPTAS)}. An EPTAS was later developed by Jansen~\cite{Ja10,Ja12}.  

In the  Bin Packing   problem, we are given a set $I$ of items, a weight function $w:I\rightarrow \mathbb{R}_{\geq 0}$ and a capacity $W > 0$. The objective is to partition the set $I$ into a minimal number of sets $S_1,\ldots, S_m$ (i.e., find a {\em packing}) such that $\sum_{i\in S_b} w(i)\leq W$ for all $1\leq b \leq m$. In \cite{KK82} the authors presented a
polynomial-time algorithm which returns a packing using $\OPT + O(\log^2 \OPT)$ bins, where $\OPT$ is the number of bins in a minimal packing. The result was later improved by Rothvo{\ss} \cite {Ro13}.

Research 
work
 on monotone submodular maximization dates back to the late 1970's. 
In \cite{NW78} Nemhauser  and  Wolsey
presented a greedy-based tight $(1-e^{-1})$-approximation for  maximizing a monotone submodular function subject to a cardinality constraint, along with a matching lower bound in the oracle model. The greedy algorithm of \cite{NW78} was later generalized to monotone submodular maximization subject to a knapsack constraint \cite{khuller1999budgeted,sviridenko2004note}.

A major breakthrough in the field of submodular optimization resulted from the 
introduction of algorithms for optimizing the {\em multilinear extension} of a submodular function (\cite{CCPV07,lee2010maximizing, CCPV11, vondrak2013symmetry,feldman2011unified,buchbinder2014submodular}).  For $\bx\in [0,1]^I$, we say that a random set $S\subseteq I$ is distributed by $\bx$ (i.e., $S\sim\bx$) if $\Pr(i\in S)=\bx_i$, and the events $(i\in S)_{i\in I}$ are independent. 
Given a function $f:2^I\rightarrow\mathbb{R}_{\geq 0}$, its  {\em multilinear extension} is $F:[0,1]^I\rightarrow\mathbb{R}_{\geq 0}$ defined as  $F(\bx)=\E_{S\sim \bx}[f(S)]$.

The input for the {\em multilinear optimization problem} is an oracle for a submodular function $f:2^I \rightarrow \mathbb{R}_{\geq 0}$ and a downward closed
solvable
polytope $P$.\footnote{
	A polytope $P\in [0,1]^I$  is {\em downward closed} if for any $\bx\in P$ and $\by\in[0,1]^I$ such that $\by\leq \bx$ (that is, $\by_i\leq \bx_i$ for every $i\in I$) it holds that $\by \in P$.
	A polytope $P\in [0,1]^I$  is {\em solvable} if, for any $\bar{\lambda} \in \mathbb{R}^I$, a point $\bx\in P$ such that $\bar{\lambda} \cdot \bx  = \max_{\by\in P }\bar{\lambda} \cdot \by$ can be computed in polynomial time, where $\bar{\lambda} \cdot \bx$ is the dot product of $\bar{\lambda}$ and $\bx$.}
The objective is to find $\bx\in P$ such that $F(\bx)$ is maximized, where $F$ is the multilinear extension of $f$. 
The problem admits 
 a $(1-e^{-1}-o(1))$-approximation 
 in the monotone case and a  $(0.385+\delta)$-approximation in the non-monotone case (for some small constant $\delta>0$) due to \cite{CCPV11} and \cite{BF19}.

Several techniques were developed for {\em rounding} a (fractional) solution for the multilinear optimization problem to an integral solution. These include Pipage 
Rounding~\cite{AS04}, Randomized Swap Rrounding~\cite{CVZ10}, and Contention Resolution Schemes~\cite{CVZ14}. These techniques led to the state of art results for many problems (e.g.,~\cite{KST13,CCPV11,AS04,CVZ10}). 


A random $(1-e^{-1}-\eps)$-approximation for the Monotone Submodular Multiple Knapsack problem was presented in~\cite{FKNRS20}. The technique 
 in \cite{FKNRS20} modifies the objective function and its domain. This modification
does not preserve submodularity of a non-montone function and the combinatorial properties of additional constraints. Thus, it does not generalize to $d$-MKCP.

A deterministic $(1-e^{-1}-\eps)$-approximation for Monotone Submodular Multiple Knapsack was later obtained by Sun et al.~\cite{SZZ20}.
Their algorithm relies on a variant of the submodular greedy of \cite{sviridenko2004note} which cannot be extended to the non-monotone case, or easily adapted to handle more than a single MKC.

 \cout{
 \ariel{consider moving to a later part in the appendix}
 The above algorithms  \cite{feldman2011unified, BF19, CCPV11} assume the polytope $P$ is given by an optimization oracle. Given $\bc \in \mathbb{R}^I$, the optimization oracle returns $\bx\in P$ such that $\bc \cdot \bx$ is maximized. 
 We note that these algorithms can be easily adapted to yield  a $0.385$ and  a $(1-e^{-1})$ approximation for the non-monotone and monotone multilinear optimization problem  given a weaker oracle. Taking $\bc \in \mathbb{R}^I$ and $\delta>0$ as input, the weaker oracle returns $\bx\in P$ such that $\bc \cdot \bx \geq (1-\delta) \max_{\by\in P} \bc \cdot \by$.  The adapted algorithms use a polynomial number of oracle queries and 
 $\delta^{-1}$ 
 is always polynomial in   the number of items. 
 We note that if $f$ is {\em modular}  then such weaker oracle can be used to compute $\bx\in P$ such that $F( \bx) \geq (1-\delta) \max_{\by\in P} F( \by)$ since, in this case, the multilinear extension is a linear function.\footnote{See Lemma~\ref{lem:multilinear_of_linear}.}
 This observation is used by the PTAS for modular $d$-MKCP. 
 We often refer to $\bx\in [0,1]^{I}$ as {\em fractional solution}.
}
 
 
 \subsection{Technical Overview}
 
 In the following we describe the technical problem solved by fractional grouping and give some insight to the way we solve this problem.  
 For simplicity, we focus on the special case of $1$-MKCP, 
 in which the number of bins is large and all bins have unit capacity. Let $(I, (w,B,W), 2^I,f)$ be a  $1$-MCKP instance  where   $W(b)=1$ for all $b\in B$.  Also, assume that no two items have the same weight. Let $S^*$ and $A$ be an optimal solution for the instance.
 
Fix an arbitrary small $\aeps>0$. We say that an item $i\in I$ is {\em heavy} if $w(i)>\aeps$; otherwise, $i$ is {\em light}. Let $H\subseteq I$ denote the heavy items. We can apply linear grouping \cite{FL81} to the heavy items in $S^*$. That is, let $h^* = |S^*\cap H|$ be the number of heavy items in $S^*$, and partition $S^*\cap H$ to $\aeps^{-2}$ groups of cardinality $ \aeps^2 \cdot h^*$,
assuming the items are sorted in decreasing order by weights
(for simplicity, assume $\aeps^{-2}$ and $ \aeps^2 \cdot h^*$ are integers). Specifically, $S^*\cap H
 = G^*_1 \cup \ldots \cup G^*_{\aeps^{-2}}$,
 where $|G^*_k|= \aeps^2 \cdot h^*$ for all $1\leq k \leq \aeps^{-2}$ and for any $i_1 \in G^*_{k_1}$, $i_2\in G^*_{k_2}$ where $k_1<k_2$ we have that $w(i_1)>w(i_2)$.  Also, for any $1\leq k\leq \aeps^{-2}$ let $q_k$, the $k$-th pivot, be the item of highest weight in $G^*_k$. 
 
 We use the pivots to generate a new collection of groups $G_1,\ldots, G_{\aeps^{-2}}$ where $G_k = \{i\in H~|~ w\left(q_{k+1}\right) < w(i) \leq w(q_k)\}$ for $1\leq k <\aeps^{-2}$,  and $G_{\aeps^{-2}}=\{i\in H~|~  w(i) \leq w(q_{\aeps^{-2}})\}$. Clearly, $G^*_k\subseteq G_k$ for any $1\leq k\leq \aeps^{-2}$. Let $X=\{i\in H~|~ w(i)> w(q_1)\}$ be the set of largest items in $H$.
 
 A standard {\em shifting} argument can be used to show that any set $S\subseteq I\setminus X$, such that $w(S)\leq |B|$ and $|S\cap G_k|\leq \aeps ^2 \cdot h^*$ for all $1\leq k \leq \aeps^{-2}$, can be packed into $(1+2\aeps)|B|+1$  bins as follows.\footnote{For a set $S\subseteq I$ we denote $w(S)=\sum_{i\in S} w(i)$.} 
 The items in $S\cap G_k$ can be packed in place of the items in $G^*_{k-1}$ in $A^*$, each of the items in $S\cap G_1$ can be packed in a separate bin (observe that $|S \cap G_1|\leq \aeps^{2} \cdot h^* \leq \aeps |B|$ as packing of $h^*$ heavy items requires at least $h^*\cdot \aeps$ bins). Finally, First-Fit can be used to pack the light items in $S$.
 
 Now, assume we know $q_1,\ldots, q_{\aeps^{-2}}$ and $h^*$; thus, the sets $G_1,\ldots,G_{\aeps^{-2}}$ and $X$ can be constructed. Consider the following optimization problem: find $S\subseteq I\setminus X$ such that $w(S)\leq |B|$, $|S\cap G_k|\leq \aeps^2 \cdot h^*$  for all $1\leq k\leq \aeps^{-2}$, and $f(S)$ is maximal. The problem is an instance of non-monotone submodular maximization with a $(1+\aeps^{-2})$-dimensional  knapsack constraint, for which there is a $(0.385-\eps)$-approximation algorithm \cite{KST13,BF19}. The algorithm can be used to find $S\subseteq I\setminus X$ which satisfies the above constraints and 
 $f(S)\geq (0.385-\eps)\cdot f(S^*)$,  as $S^*$ is a feasible solution for the problem. Subsequently, $S$ can be packed into bins using a standard bin packing algorithm. This will lead to a packing of $S$ into roughly 
 $(1+2\aeps)|B| +O(\log^2|B|)$ bins. By removing the bins of least value (along with their items), and using the assumption that $|B|$ is sufficiently large, we can obtain a set $S'$ and an assignment of $S'$ into $B$ such that $f(S)$ is arbitrarily close to $0.385 \cdot f(S^*)$. 
 
 Indeed, we do not know the values of  $q_1,\ldots, q_{\aeps^{-2}}$ and $h^*$. This prevents us from  using the above approach. However, as in~\cite{BEK16}, we can overcome this difficulty through  exhaustive enumeration. Each of   $q_1,\ldots, q_{\aeps^{-2}}$ and $h^*$ takes one of $|I|$ possible values. Thus, by iterating over all $|I|^{1+\aeps^{-2}}$ possible values for $q_1,\ldots, q_{\aeps^{-2}}$ and $h^*$, and solving the above problem for each, we can find a solution of value at least
 $0.385 \cdot f(S^*)$.
 
 While this approach is useful for our restricted class of instances, 
 due to the use of exhaustive enumeration it does not scale
 to general instances, where bin capacities may be arbitrary. Known techniques (\cite{FKNRS20}) can be used to reduce the number of unique bin capacities in a general MKC to be logarithmic in $|B|$. As enumeration is required for each unique capacity, this results in $|I|^{\Theta(\log |B|) }$ iterations, which is non-polynomial. 
 
 Fractional Grouping overcomes this hurdle by using a polytope $P\subseteq [0,1]^I$ to represent an MKC. A grouping  
 $G^{\by}_1,\ldots, G^{\by}_{\tau}$ with $\tau\leq \aeps^{-2}+1$ is derived from a vector $\by\in P$. The polytope $P$ bears some similarity to configuration linear programs used in previous works (\cite{Ja12,FGMS11,BEK16}). While $P$ is not solvable, it satisfies an approximate version of solvability which suffices for our needs.
 
 Fractional grouping satisfies the main properties of the grouping  defined for $S^*$. Each of the groups contains roughly the same number of fractionally selected items. That is,
 $\sum_{i\in G^{\by}_k} \by_i\approx \aeps^2 |B|$ for all $1\leq k \leq \tau$.  Furthermore, we show that if $\by$ is strictly contained in $P$ then any subset $S\subseteq I$ satisfying $(i)$ $|S\cap G_k|\leq \aeps |B|$ for all $1\leq k \leq \tau$, and $(ii)$ $w(S\setminus H)$ is sufficiently small, can be packed into strictly less than $|B|$ bins (see the details in Section \ref{sec:grouping}). 
The existence of a packing for $S$ relies on a shifting argument similar to the one used above. In this case, however, the structure of the polytope $P$ replaces the role of $S^*$ in our discussion.

 This suggests the following algorithm. Use the algorithm of \cite{BF19} to find $\by \in P$ such that $F(\by)\geq (0.385-\eps) f(S^*)$, and sample a random set $R\sim (1-\delta)^2 \by$. By the above property, $R$ can be packed into strictly less than $|B|$ bins with high probability, as $\E\left[|R\cap G_k|\right]\ll \aeps |B|$.
 Thus, $R$ can be packed into $B$ using a bin packing algorithm.  Standard submodular bounds also guarantee that $\E[f(R)]$ is arbitrarily close to $F(\by)$. Hence, we can obtain an approximation ratio arbitrarily close to $(0.385-\eps)$ while avoiding enumeration.

 This core idea of fractional grouping for bins of uniform capacities can be scaled to obtain Theorem~\ref{thm:main}. This scaling involves use of existing techniques for submodoular optimization (\cite{FKNRS20, CVZ10,CVZ11,CCPV11, BF19}), along with a novel {\em block association} technique we apply to handle MKCs with arbitrary bin capacities.

 \noindent{\bf Organization.}
 We present the fractional grouping technique in Section~\ref{sec:grouping}. Our algorithms for uniform bin capacities and the general case are given in Section~\ref{sec:block} and~\ref{sec:algorithm}, respectively. Due to space constraints, the block association technique is presented in Appendix~\ref{sec:association}.

\cout{
Using the properties of fractional grouping, it can be shown that $R\cap I_j$ can be packed

 We note that the lemma provides a solution for the problem mentioned above. Let $(\bx,\by^0,\ldots, \by^\ell)\in P^e$, the extended $\gamma$-partition polytope of $(w,B,W)$ and $(K_j)_{j=0}^{\ell}$. Also, let $R\sim (1-\delta)^2 \cdot \bx$ for $\delta = 4\aeps$. We can evaluate the association $(I_j)_{j=0}^{\ell}$ of $(1-\delta)\cdot (\bx,\by^0,\ldots, \by^{\ell})$ and use a bin packing algorithm to pack $R\cap I_j$ into bins of capacity $W^*_{K_j}$ for $0\leq j\leq \ell$ with $|K_j|>1$. It can be shown that with high probability $R\cap I_j$ satisfies the conditions of Lemma \ref{lem:grouping} for every $j$ with $|K_j|>1$. Thus the bin packing algorithm will return a packing which uses no more than $|K_j|$ bins for any $j$ with high probability.
}

\section{Fractional Grouping}
\label{sec:grouping}

Given an MKC $(w,B,W)$ over $I$, a subset of bins $K\subseteq B$ is a {\em block} if all the bins in $K$ have the same capacity. 
Denote by $W^*_K$ the capacities of the bins in block $K$, then $W^*_K=W(b)$ for any $b\in K$.

We first define a polytope $P_K$ which represents the block $K\subseteq B$ of an  MKC $(w,B,W)$ over $I$. To simplify the presentation, we assume the MKC $(w,B,W)$ and $K$ are fixed  throughout this section. W.l.o.g., assume that $I=\{1,2,\ldots, n\}$ and $w(1)\geq w(2)\geq \ldots \geq w(n)$. 
A {\em$K$-configuration} is a subset $C\subseteq I$ of items which fits into a single bin of block $K$, i.e., $w(C)\leq W^*_K$. We use $\cC_K$ to denote the set of all $K$-configurations. Formally,
$\cC_K = \left\{C\subseteq I~ \middle|~ w(C)\leq W^*_K \right\}$. 

\begin{definition}
	\label{def:extended_block}
	The {\em extended block polytope of $K$} is
	\begin{equation}
		\label{eq:PeK_def}
		P^e_K = \left\{\by\in [0,1]^I, \bz\in [0,1]^{\cC_K} \middle|
		\begin{array}{lcrl}
			& &\displaystyle \sum_{C\in \cC_K} \bz_C ~&\leq~ |K| \\
			\forall i\in I: & &
			\displaystyle
			\by_i ~&\leq ~
			\displaystyle
			\sum_{C\in \cC_K \mbox{ s.t. } i\in C} \bz_C 
		\end{array}
		\right\}
	\end{equation}
\end{definition}	
The first constraint in~(\ref{eq:PeK_def}) bounds the number of selected configurations by the number of bins. The second constraint requires that
each selected item is (fractionally) covered by a corresponding set of configurations. It is easy to verify that, for any $(\by, \bz)\in P^e_K$, it holds that 
$\sum_{i\in I} w(i) \cdot \by_i \leq |K|  \cdot W^*_K$. 

\begin{definition}
\label{def:block_polytope}
The {\em block polytope of $K$} is
\begin{equation}
\label{eq:P_def}
P_K=\left\{\by\in [0,1]^I | ~ \exists \bz\in [0,1]^{\cC_k}: (\by,\bz)\in P^e_K \right\}.
\end{equation} 
\end{definition}


While $P^e_K$ and $P_K$ are defined by an exponential number of variables, it follows from standard arguments (see, e.g., \cite{FGMS11, KK82}) that, for any $\bc\in \mathbb{R}^I$, $\max_{\by\in P_K} \bc \cdot \by$ can be approximated.
\begin{lemma}
	\label{lem:poly_fptas}
	There is a {\em fully polynomial-time approximation scheme (FPTAS)} for the problem of finding $\by \in P_K$ such that $\bc \cdot \by$ is maximal, given an $MKC$ $(w,B,W)$, a block $K\subseteq B$ and  a vector $\bc\in \mathbb{R}^I$, where $P_K$ is the  block polytope of $K$. 
\end{lemma}
A formal proof is given in  Appendix~\ref{app:poly_fptas} as a special case of Lemma~\ref{lem:instance_fptas}.
We say that $A:K\rightarrow 2^I$ is a {\em feasible assignment for $K$} if  $w(A(b))\leq W^*_K$ for any $b\in K$.
Also, we use $\one_S=\bx\in \{0,1\}^I$, where $\bx_i = 1$ if $i\in S$ and $\bx_i = 0$ if $i\in I\setminus S$.
The next lemma implies that the definition of $P^e_K$ is sound for the problem.  
\begin{lemma}
	\label{lem:soundness}
	Let $A$ be a feasible assignment for $K$ and $S=\bigcup_{b\in K} A(b)$. Then $\one_S \in P_K$. 
\end{lemma}
The lemma is easily proved, by setting $\bz_C=1$ if $A(b)=C$ for some $b \in B$, 
and $\bz_C=0$ otherwise.
We say an item $i\in I$ is {\em $\aeps$-heavy} for $\aeps>0$ (w.r.t $K$) if $W^*_K \geq w(i) > \aeps \cdot W^*_K$; otherwise, $i\in I$ is {\em $\aeps$-light}.
Denote by $H_{K,\aeps}$ and $L_{K, \aeps}$ the sets of $\aeps$-heavy items and $\aeps$-light items, respectively.

Given a vector $\by \in P_K$, we now describe the partition of $\aeps$-heavy items into groups $G_1, \ldots , G_\tau$, for some $\tau \leq \aeps^{-2} +1$.
Starting with $k=1$ and $G_k=\emptyset$, add items from $H_{K,\aeps}$ to the current group $G_k$ until $\sum_{i\in G_k} \by_i \geq \aeps |K|$. Once the constraint is met, mark the index of the last item in $G_k$ as $q_k$, the $\aeps$-{\em pivot} of $G_k$, close $G_k$ and open a new group, $G_{k+1}$. 
Each of the groups $G_1, \ldots ,G_{\tau-1}$ represents a fractional  selection of $\approx \aeps|K|$ heavy items of $\by$. The last group, $G_\tau$, contains the remaining items in $H_{K,\aeps}$, for which the $\aeps$-pivot is $q_{max}$ (last item in $H_{K,\aeps}$).
We now define formally the partition process.
\begin{definition}
	\label{def:partition}
	Let $\by \in P_K$ and $\aeps\in\left(0,\frac{1}{2}\right]$.  Also, let $q_0 \in \{0,1,\ldots, n\}$ and $q_{\max}\in I$ such that $H_{K,\aeps}= \{ i\in I ~|~ q_0 < i \leq q_{\max }\}$. 
	The $\aeps$-pivots of $\by$, given by $q_1, \ldots, q_{\tau}$, are 
	defined inductively, i.e.,
	$$q_k = \min \left\{s\in H_{K,\aeps} ~\middle|~ \sum_{i=q_{k-1}+1}^s \by_i \geq \aeps\cdot |K|  \right\}.$$
	If the set over which the minimum is taken is empty, let $\tau=k$ and $q_{\tau}=q_{\max}$. 
	The {\em $\aeps$-grouping} of $\by$ consists of the sets $G_1,\ldots, G_\tau$, where $G_k =\left\{i\in H_{K,\aeps} ~\middle|~ q_{k-1}<i \leq q_k\right\}$ for $1\leq k \leq \tau$. 
\end{definition}

Given a polytope $P$ and $\delta\in \mathbb{R}$, we use the notation $\delta P = \{\delta \bx ~|~ \bx \in P\}$. The main properties of fractional grouping are summarized in the next lemma. 
\begin{lemma}[Fractional Grouping]
	\label{lem:grouping}
	For any $\by \in P_K$ and $0<\aeps<\frac{1}{2}$ there is a polynomial time algorithm which computes a partition $G_1,\ldots, G_{\tau}$ of $H_{K,\aeps}$  with $\tau\leq \aeps^{-2}+1$ for which the following hold:
	\begin{enumerate}
		\item  $\sum_{i\in G_k} \by_i \leq \aeps \cdot |K|+1$ for any $1\leq k \leq \tau$.
		\item 
		Let $S\subseteq H_{K,\aeps} \cup L_{K,\aeps}$ such that $|S\cap G_k |\leq    \aeps |K|$ for every $1\leq k \leq \tau$, and
		$w(S\cap L_{K,\aeps})\leq \sum_{i\in L_{K,\aeps}} \by_i \cdot w(i) + \lambda \cdot W^*_K$ for some $\lambda\geq 0$. Also, assume  $\by \in (1-\delta) P_K$ for some $\delta \geq 0$. Then $S$ can be packed into $(1-\delta +3\aeps)|K| +   4\cdot 4^{\aeps^{-2}} + 2 \lambda$ bins of capacity $W^*_K$. 
		\label{prop:grouping_third}
	\end{enumerate}
	We refer to $G_1,\ldots, G_{\tau}$ as {\em the $\aeps$-grouping} of $\by$. 
\end{lemma}
\begin{proof}
It follows from Definition \ref{def:partition} that $G_1,\ldots, G_{\tau}$ can be computed in polynomial time. Also, $\sum_{i\in G_\tau} \by_i< \aeps \cdot |K| $  and
 \begin{equation}
\label{eq:group_size}
\forall 1\leq k <\tau:~~~~~\aeps\cdot   |K| \leq \sum_{i\in G_k}\by_i \leq \aeps \cdot |K| +1.
\end{equation}  Furthermore, $\tau\leq \aeps^{-2}+1$. Thus, it remains to show Property~\ref{prop:grouping_third} in the lemma.

Define the {\em type} of a configuration $C\in \cC_K$, denoted by $\type(C)$, as  the vector $T\in \mathbb{N}^\tau$ with $T_k = | C \cap G_k|$.
 Let $\cT =\{\type(C)~|~ C\in \cC_K\} $  be the set of all  types. 
 Given a type $T \in \cT$, consider a set of items $Q\subseteq H_{K,\aeps}\setminus G_1$, such that $|Q\cap G_k|\leq T_{k-1}$ for any $2\leq k\leq \tau$, then
 $w(Q)\leq W^*_K$. This is true since we assume the items in $H_{K,\aeps}$ are sorted in non-increasing order by weights. We use this key property to construct a packing for $S$.
 
 We note that $\sum_{k=1}^{\tau}|C\cap G_k| < \aeps^{-1}$ for any $C\in \cC_K$  (otherwise $w(C)>W^*_K$, as $G_k\subseteq H_{K,\aeps}$). It follows that $|\cT|\leq 4^{\aeps^{-2}}$. Indeed, the number of types is bounded by the number of different non-negative integer $\tau$-tuples whose sum is at most $\aeps^{-1}$.

By Definition~\ref{def:extended_block}, there exists $\bz\in [0,1]^{\cC_K}$ such that $(\by,\bz)\in (1-\delta)P^e_K$. For $T\in \cT$, let ${\ett} =
\sum_{ C\in \cC_K \textnormal{ s.t. } \type(C)=T }
 \bz_C$. Then, for any  $1\leq k\leq \tau-1$, we have
\begin{equation}
\label{eq:eta_bound}
\aeps |K|\leq  \sum_{i\in G_k} \by_i \leq \sum_{i\in G_k} \sum_{~C\in \cC_k  \text{ s.t. } i\in C~}
 \bz_C = \sum_{C\in \cC_K} |G_k\cap C| \bz_C=
 \sum_{T\in \cT} T_k \cdot {\ett}
\end{equation}
The first inequality follows  from \eqref{eq:group_size}. The second inequality follows from \eqref{eq:PeK_def}. The two equalities follow by rearranging the terms. 

Using $\bz$ (through the values of ${\ett}$) we define an assignment of 
$S\cap (G_2 \cup \ldots\cup  G_\tau)$ to $\eta = \sum_{T\in \cT} \ceil{\ett}$ bins.  We initialize $\eta$ sets (bins) $A_1,\ldots, A_\eta = \emptyset $ and associate a type with each set $A_b$, such that there are $\ceil{\ett}$ sets associated with the type $T\in \cT$, using a function $R$. That is, let $R:\{1,2,\ldots, \eta\}\rightarrow \cT$  such that $|R^{-1}(T)| = \ceil{\ett}$. 
We assign
the items in $S\cap (G_2 \cup \ldots \cup G_\tau)$ to $A_1,\ldots, A_{\eta}$  while ensuring that $|A_b\cap G_k|\leq R(b)_{k-1}$ for any $1\leq b \leq \eta$ and $2\leq k \leq \tau$. In other words,
the number of items assigned to $A_b$ from $G_k$ is at most the number of items from $G_{k-1}$ in the configuration type $T$ assigned to bin $b$ by $R$. The assignment is obtained as follows.
For every $2\leq k\leq \tau$, iterate over the items $i\in S\cap G_k$, find $1\leq b \leq \eta$ such that $|A_b\cap G_k| < R(b)_{k-1}$  and set $A_b \leftarrow A_b \cup \{i\}$. It follows from \eqref{eq:eta_bound} and the conditions of the lemma that such $b$ will always be found. 

Upon completion of the process, we have that $S\cap \left( G_2\cup \ldots  \cup G_\tau\right) = A_1\cup \ldots \cup A_{\eta}$.
Furthermore, for every $1\leq b\leq \eta$, there are $C\in \cC_K$ and $T\in \cT$ such that $\type(C)=T=R(b)$.  Since $A_b \subseteq G_2\cup \ldots \cup G_{\tau}$, we have
$$
w(A_b) = \sum_{k=2}^{\tau} w(A_b\cap G_k ) \leq \sum_{k=2}^{\tau}   T_{k-1} \cdot w(q_{k-1}) =
\sum_{k=2}^{\tau}   |C\cap G_{k-1}| \cdot w(q_{k-1}) 
\leq \sum_{i\in C} w(i)\leq W^*_K.
$$
The first inequality holds since $w(q_{k-1})\geq w(i)$ for every $i\in G_{k}$, and the second holds since $w(q_{k-1})\leq w(i)$ for every $i\in G_{k-1}$. By similar arguments,
for every $2\leq k \leq \tau$, we have
\begin{equation}
\label{eq:group_weight}
w(S\cap G_k  )\leq |S\cap G_k| \cdot w(q_{k-1}) \leq \aeps |K| \cdot w(q_{k-1}) \leq \sum_{i\in G_{k-1} } \by_i \cdot w(q_{k-1})\leq 
\sum_{i\in G_{k-1} } \by_i\cdot  w(i).
\end{equation}
The third inequality is due to~\eqref{eq:group_size}. 
 Using \eqref{eq:group_weight} and the conditions in the lemma, 
\begin{equation}
\label{eq:no_G1_bound}
\begin{aligned}
w\left(S\setminus G_1\right) &= w(S\cap L_{K, \aeps} )+ \sum_{k=2}^{\tau }w(S\cap G_k) \
\leq \sum_{i\in L_{K,\aeps} }\by_i  w(i) +\lambda W^*_K +\sum_{k=1}^{\tau-1} \sum_{i\in G_{k}} \by_i  w(i)\\
&
\leq  \sum_{i\in I } \by_i \cdot w(i) + \lambda W^*_K\leq (1-\delta)W^*_K \cdot |K|+\lambda W^*_K.
\end{aligned}
\end{equation}

We use First-Fit (see, e.g., Chapter 9 in \cite{Va13}) to add the items in  $S\cap L_{K,\aeps}$ to the sets (=bins) $A_1, \ldots, A_{\eta}$ while maintaining the capacity constraint, $w(A_b)\leq W^*_K$. 
%
First-Fit
iterates over the items
$i\in S\cap L_{K,\aeps}$ and searches for a minimal $b$ such that $w(A_b \cup \{i\})\leq W^*_K$. If such $b$ exists, First-Fit updates $A_b \leftarrow A_b\cup \{i\}$; otherwise, it adds a new bin with $i$ as its content.  
%
Let $\eta'$ be the number of bins by the end of the process. 
 As $w(i)\leq \aeps W^*_K$ for $i\in S\cap L_{K,\aeps}$, and due to~\eqref{eq:no_G1_bound}, it holds that 
 $\eta' \leq \max\{\eta,   \left(|K|(1-\delta)+\lambda\right)(1+2\aeps)+1\}$. 
Finally,
$$
\eta = \sum_{T\in \cT} \ceil{\ett}
\leq |\cT| + \sum_{T\in \cT} 
{\ett}\leq
4^{\aeps^{-2}}+ \sum_{C\in \cC_K} \bz_C \leq 4^{\aeps^{-2}}+ (1-\delta )|K|.
$$
Thus, there is a packing of $S\setminus {G_1}$ into at most $(1-\delta)|K| + 4^{\aeps^{-2}}+1  + 2\aeps|K|+2\lambda$ bins of capacity $W^*_K$. Since $|S\cap G_1|\leq \aeps |K|$, each of the items in $S\cap G_1$ can be packed into a bin of its own. This yields a packing using at most $(1-\delta+3\aeps)|K| + 4\cdot4^{\aeps^{-2}}+2\lambda$  bins.

\end{proof}

\section{Uniform Capacities }
\label{sec:block}

In this section we apply fractional grouping (as stated in Lemma~\ref{lem:grouping}) to solve the Monotone Submodular Multiple Knapsack Problem with Uniform Capacities (USMKP). 
An instance of the problem consists of an MKC $(w,B,W)$  over a set $I$ of items, such that $W^*_B=W(b)$ for all $b\in B$,  and  a submodular function $f:2^I\rightarrow \mathbb{R}_{\geq 0}$.
For simplicity, we associate a solution for the problem with a feasible assignment $A:B\rightarrow 2^I$.
Then, the set of assigned items is given by $S=\bigcup_{b\in B} A(b)$.

Our algorithm for USMKP instances applies
{\em Pipage Rounding}~\cite{AS04,CCPV07}. 
Given  a (fractional) solution $\bx\in [0,1]^I$, a submodular function $f$, and two items $i_1, i_2 \in I$ with costs $c_1, c_2$, let $F$ be the multilinear extension of $f$. 
Pipage Rounding returns  a new random solution $\bx'\in [0,1]^I$ such that $\E\left[F(\bx')\right]\geq F(\bx)$, $\bx'_i=\bx_i$ for $i\in I\setminus \{i_1, i_2\}$, $\bx_{i_1} \cdot c_1 +\bx_{i_2} \cdot c_2 =\bx'_{i_1} \cdot c_1 +\bx'_{i_2} \cdot c_2 $, and either $\bx'_{i_1}\in \{0,1\}$ or $\bx'_{i_2}\in \{0,1\}$. 
Algorithm~\ref{alg:uniform} calls the subroutine $\pipage(\bx, f, G, \bc)$, which can be implemented by an iterative application of Pipage rounding, as summarized in the next result.
\begin{lemma}
	\label{lem:pipage}
	There is a polynomial time procedure $\pipage(\bx, f, G, \bc)$ for which the following holds. Given $\bx\in [0,1]^{I}$, a submodular function $f:2^I\rightarrow \mathbb{R}_{\geq 0}$, a subset of items $G\subseteq I$ and a cost vector for the items $\bc \in \mathbb{R}_{\geq 0 }^{G}$, the procedure returns a random vector $\bx'\in [0,1]^I$ such that $\E\left[F(\bx')\right]\geq F(\bx)$, $\bx'_i\in \{0,1\}$ for $i\in G$, $\bx'_i = \bx_i$ for all $i\in I\setminus G$, and there is $i^*\in G$ such that $\sum_{i\in G} \bx'_i \cdot c_i \leq c_{i^*}+\sum_{i\in G} \bx_i \cdot c_i$.  
\end{lemma}

To solve USMKP instances, our algorithm initially
finds $\by\in P_B$, where $P_B$  is the block polytope of $B$ (note that $B$ is a block in this case), for which $F(\by)$ is {\em large} ($F$ is the multilinear extension of the value function $f$). 
The algorithm chooses a small value for $\aeps$ and uses $G_1,\ldots, G_{\tau}$, the $\aeps$-grouping of $(1-4\aeps)\by$, to guide the rounding process. 
Pipage rounding is used to convert $(1-4\aeps)\cdot\by$ to $S\subseteq I$ while preserving the number of selected items from each group as $\approx \aeps |B|$, and the total weight of items selected from $L_{B,\aeps}$ (i.e., $\aeps$-light items) as $\approx (1-4\aeps)\cdot \sum_{i\in L_{B,\aeps}} \by_i  \cdot w(i)$.
An approximation algorithm for bin packing is then used to find a packing of $S$ to the bins. Lemma \ref{lem:grouping} ensures the resulting packing uses at most $|B|$ bins for sufficiently large $B$. In case the packing 
requires more than $|B|$ bins we simply assume the algorithm returns an empty solution.
We give the pseudocode in Algorithm~\ref{alg:uniform}.

\begin{algorithm}
	\SetAlgoLined
	\SetKwInOut{Input}{Input}\SetKwInOut{Configuration}{Configuration}
	\DontPrintSemicolon
	
	\KwIn{An MKC $(w,B,W)$ over $I$ with uniform capacities. A submodular function $f:2^I\rightarrow \mathbb{R}_{\geq 0}$.}
	
	Find an approximate solution $\by\in P_B$ for $\max_{\by \in P_B} F(\by)$, where $P_B$ is the block polytope of $B$, and $F$ is the multilinear extension of $f$. 
	\label{uniform:opt}\;
	Choose $\aeps = \min\left\{
	\left( \log|B|\right)^{-\frac{1}{4}},  \frac{1}{2}\right\}
	$.\;
	Set		$\by^0\leftarrow (1-4\aeps)\by$. and let $G_1, \dots, G_{\tau}$ be the $\aeps$-grouping of $\by^0$. 
	\label{def_x_for_pipage}\;
	

	{\bf for} {$k=1,2,\ldots, \tau$}  {\bf do}
		$\by^k \leftarrow \pipage\left(\by^{k-1},f,G_k, \bar{1}\right)$. \;
	$\by' = \pipage\left(\by^{\tau},f,L_{B,\aeps}, \left(w(i) \right)_{i\in L_{B, \aeps}}\right)$.
	\label{step:def_yprime}\;
	Let $S=\{i\in I ~|~ \by'_i = 1\}$. \label{uniform:S_def}\;

	Pack the items in $S$ into $B$ using a bin packing algorithm. Return the resulting assignment. 
	\label{uniform:pack}
	
	
	\caption{Submodular Multiple Knapsack with Uniform Capacities} 
	\label{alg:uniform}
\end{algorithm}

\begin{lemma}
Algorithm \ref{alg:uniform} yields a $\left(1-e^{-1}-O\left(\left(\log |B|\right)^{-\frac{1}{4}} \right)\right)$-approximation for  USMKP.
\end{lemma}
\begin{proof}
	Let $A^*$ be an optimal solution for the input instance, and $\OPT = f\left( \bigcup_{b\in B} A^*(b) \right)$ its value. By Lemma \ref{lem:soundness}, $\one_{\bigcup_{b\in B} A^*(b)}  \in P_B$. 
Let $c= 1-e^{-1}$.
  By using the algorithm of \cite{CCPV11}
  we have that
$F(\by)\geq \left(c-\frac{1}{|I|}\right)\cdot \OPT$ 
($\by$ is defined in Step~\ref{uniform:opt} of Algorithm~\ref{alg:uniform}).  
The algorithm of \cite{CCPV11}  is used with the FPTAS of Lemma \ref{lem:poly_fptas} as an oracle for solving linear optimization  problems over $P_B$.
We note that this would not affect the approximation guarantee.

Since the multilinear extension has negative second derivatives  \cite{CCPV11}, it follows that $F(\by^0) \geq (1-4 \aeps) \cdot \left(c-\frac{1}{|I|}\right) \cdot \OPT$. 
Now, consider the vector $\by'$ output in Step~\ref{step:def_yprime} of the algorithm.
By Lemma \ref{lem:pipage}, it follows that $\E\left[F(\by')\right] \geq F(\by^0) \geq (1-4\aeps) \cdot \left(c-\frac{1}{|I|}\right) \cdot \OPT$, and $\by'\in \{0,1\}^{I}$ (note that $\by'_i=\by_i =0$ for any $i$ with $w(i)>W^*_B$ due to \eqref{eq:PeK_def}). Thus, for the set $S$ defined in Step~\ref{uniform:S_def} of the algorithm, we have 
$\E\left[f(S)\right] \geq (1-4\aeps) \cdot \left(c -\frac{1}{|I|}\right) \cdot \OPT\geq
\left(c - O\left( \left(\log |B|\right)^{-\frac{1}{4}}\right)\right) \cdot \OPT$ (observe we may assume w.l.o.g that $|I|\geq |B|$). 

To complete the proof, it remains to show that the bin packing algorithm in Step~\ref{uniform:pack}  packs all items in $S$ into the bins $B$. By Lemma \ref{lem:pipage}, for any $1\leq k \leq \tau$, it holds that $|S\cap G_k| = \sum_{i\in G_k} \by'_i \leq 1+ \sum_{i\in G_k} \by^0_i\leq \aeps \cdot |B| +2$ (the last inequality follows from Lemma~\ref{lem:grouping}).
Similarly, there is  $i^*\in L_{B,\aeps}$ such that 
$$
w(S\cap L_{B,\aeps})= \sum_{i\in L_{B,\aeps}} \by'_i \cdot w(i)
\leq
w(i^*) +\sum_{i\in L_{B,\aeps}} \by^0_i\cdot w(i)
\leq
  \aeps \cdot W^*_B  +\sum_{i\in L_{B,\aeps}} \by^0_i\cdot w(i).$$
  
To meet the conditions of Lemma \ref{lem:grouping}, we need to remove (up to) two items from each group, i.e., $S \cap G_k$, for $1 \leq k \leq \tau$. Let $R\subseteq S$ be a minimal subset such that $|(S \setminus R) \cap G_k|\leq \aeps |B|$ for all $1\leq k \leq \tau$. By the above we have that $|R|\leq 2\cdot \tau \leq 2 \cdot (\aeps^{-2}+1)$. 
Therefore, $S\setminus R$ satisfies the conditions of Lemma \ref{lem:grouping}. Hence, by taking $\delta = 4\aeps$ and $\lambda =\aeps$, the items in $S\setminus R$ can be packed into $(1-\aeps)|B| + 4\cdot 4^{\aeps^{-2}} +2\aeps $ bins. By using an additional bin for each item in $R$, and assuming $|B|$ is large enough,  the items in  $S$ can be packed into 
$$ (1-\aeps)|B| + 4\cdot 4^{\aeps^{-2}} +2\aeps +2 \cdot (\aeps^{-2}+1)   \leq 
|B| -\frac{|B|}{ \left( \log |B| \right)^{\frac{1}{4}}}
+ 
5\cdot 4^{\sqrt{\log |B|}} +3
\leq |B|
 $$
  bins of capacity $W^*_B$. Recall that the algorithm of~\cite{KK82} returns a packing in at most $\OPT+O(\log^2\OPT)$ bins. Thus, for  large enough $|B|$, the number of bins used in Step~\ref{uniform:pack} of Algorithm~\ref{alg:uniform} is at most
$$|B| -\frac{|B|}{ \left( \log |B| \right)^{\frac{1}{4}}}
	+   5\cdot 4^{\sqrt{\log |B| }}
	\ +O(\log^2|B|)\leq |B|.$$

Finally, we note that Algorithm~\ref{alg:uniform} can be implemented in polynomial time. 

\end{proof}

\section{Approximation Algorithm}
\label{sec:algorithm}

In this section we present our algorithm for general instances of $d$-MKCP, which gives the result in Theorem \ref{thm:main}. 
In designing the algorithm, a key observation is that
we can restrict our attention to $d$-MKCP instances of certain structure, with other
crucial properties satisfied by the objective function.
For the {\em structure}, we assume the bins are partitioned into {\em levels} by capacities, using the following definition of~\cite{FKNRS20}. 

\begin{definition}
	For any $N\in \mathbb{N}$, a set of bins $B$ and capacities $W:B\rightarrow \mathbb{R}_{\geq 0}$,  
	a partition $(K_j)_{j=0}^\ell$ of $B$
	is {\em $N$-leveled} if, for all  $0\leq j \leq \ell$, $K_j$ is a block and  $|K_j|= N^{\floor{\frac{j}{N^2}}}$. 
	We say that $B$ and $W$ are {\em $N$-leveled} if such a partition exists.
\end{definition}

For $N,\xi\in \mathbb{N}$,   {\em $(N,\xi)$-restricted  $d$-MKCP} is the special case of  $d$-MKCP in which for any instance ${\cal R}=\left(I,\left(w_t,B_t,W_t\right)_{t=1}^d,\II, f\right)$ it holds that $B_t$ and $W_t$ are $N$-leveled for all $1\leq t\leq d$, and $f(\{i\})-f(\emptyset)\leq \frac{\OPT}{\xi}$ for any $i\in I$, where $\OPT$ is the value of an optimal solution for the instance. We assume the input for $(N,\xi)$-restricted $d$-MKCP includes the $N$-leveled partition
$(K^t_j)_{j=0}^{\ell_t}$ of $B_t$ for all $1\leq t\leq d$.  
Combining standard enumeration with the structuring technique of \cite{FKNRS20}, we derive  the next result, whose proof is given in Appendix~\ref{app:restricted}.
\begin{lemma}
	\label{lem:restricted}
	For any $N,\xi, d \in \mathbb{N}$ and $c\in [0,1]$, a polynomial time  $c$-approximation for modular/ monotone/ non-monotone 
	$(N,\xi)$-restricted  $d$-MKCP with a matroid/ matroid intersection/ matching/ no additional constraint  implies a polynomial time   $c\cdot \left(1-\frac{d}{N}\right)$-approximation for $d$-MKCP, with the same type of function and same type of additional constraint.
\end{lemma}

Our algorithm for $(N,\xi)$-restricted  $d$-MKCP associates a polytope with each instance.
To this end, we first generalize the definition of a block polytope (Definition \ref{def:block_polytope}) to represent an MKC.
We then use it to define a polytope for the whole instance. 
\begin{definition}
	\label{def:extended_gamma_partition}
	For $\gamma>0$, the  extended $\gamma$-partition polytope of an MKC $(w,B,W)$ and the partition  $\left(K_j\right)_{j=0}^{\ell}$ of $B$ to blocks is
	\begin{equation}
		\label{eq:ex_gamma_partition}
		P^e= \left\{ (\bx, \by^0, \ldots, \by^\ell)~\middle|~ \begin{array}{lcc} &\bx \in [0,1]^I&~~ \\  \forall  0\leq j \leq \ell:& \by^{j}\in P_{K_j} &\\ & \sum_{j=0}^{\ell} \by^j = \bx&\\
			\forall 0\leq j \leq \ell, |K_j|=1, i\in I \setminus L_{K_j, \gamma }:& ~\by^j_i=0&
		\end{array}
		\right\}
	\end{equation}
	where $P_{K_j}$ is the block polytope of $K_j$, 
	and $L_{K_j,\gamma}$ is the set of $\gamma$-light items of $K_j$.
	The $\gamma$-partition polytope of  $(w,B,W)$ and $\left(K_j\right)_{j=0}^{\ell}$ is
	\begin{equation}
		\label{eq:gamma_partition}
		P= \left\{ \bx  \in [0,1]^I~\middle|~ \exists \by^0,\ldots \by^{\ell} \in [0,1]^I \text{ s.t. } (\bx, \by^0, \ldots, \by^\ell) \in P^e~\right\}
	\end{equation} 
\end{definition}
	The last constraint in (\ref{eq:ex_gamma_partition}) forbids the assignment of  $\gamma$-heavy items to blocks of a single bin. This technical requirement is  used to show a concentration bound.

Finally, the {\em $\gamma$-instance polytope} of $\left(I,\left(w_t,B_t,W_t\right)_{t=1}^d,\II, f\right)$ and a partition  $\left(K^t_j\right)_{j=0}^{\ell_t}$ of $B_t$ to blocks, for $1\leq t\leq d$, is $P=P(\II) \cap \left( \bigcap_{t=1}^d P_t\right)$, where $P(\II)$ is the convex hull of $\II$ and 
 $P_t$ is the $\gamma$-partition polytope of $(w_t, B_t, W_t)$ and $\left(K^t_j\right)_{j=0}^{\ell_t}$.
In the {\em instance polytope optimization problem}, we are given 
a $d$-MKCP instance $\mathcal{R}$ with a partition of the bins to blocks for each MKC, $\bc\in \mathbb{R}^I$ and  $\gamma>0$. The objective is to find $\bx \in P$ such that $\bx \cdot \bc$ is maximized, where  $P$ is the $\gamma$-instance polytope of $\mathcal{R}$. While the problem cannot be solved exactly, it admits an FPTAS.
\begin{lemma}
\label{lem:instance_fptas}
There is an FPTAS for the instance polytope optimization problem.	
\end{lemma}
The lemma follows from known techniques for approximating an exponential size linear program using an approximate separation oracle for the dual program. We give the proof in Appendix~\ref{app:poly_fptas}. 

 The next lemma asserts that the $\gamma$-instance polytope provides an approximate 
representation for the instance as a polytope.
\begin{lemma}
	\label{lem:instance_soundness}
	Given an $(N,\xi)$-restricted  $d$-MKCP instance $\mathcal{R}$ with objective function $f$, let $S, (A_t)_{t=1}^d$ be an optimal solution for $\mathcal{R}$ and $\gamma>0$. Then there is $S'\subseteq S$ such that 
	$\one_{S'} \in P$ and $f(S')\geq \left(1-\frac{N^2\cdot d}{\xi\cdot \gamma }\right)f(S)$, where $P$ is the $\gamma$-instance polytope of $\mathcal{R}$.
\end{lemma}
Lemma \ref{lem:instance_soundness}  is proved constructively by removing the $\gamma$-heavy items assigned to blocks of a single bin in $A_t$, for $1\leq t\leq d$.  We give the proof in Appendix~\ref{app:algorithm}.

Recall that $F$ is the multiliear extension of the objective function $f$.
Our algorithm finds a vector $\bx$ in the instance polytope for which $F(x)$ approximates the optimum. 
The fractional solution $\bx$ is then rounded to an integral solution. 
Initially, a random set $R\in \II$  is sampled, with $\Pr(i\in R) = (1-\delta)^2 \bx_i$.\footnote{Recall that $\II$ is the additional constraint.}
The technique by which $R$ is sampled depends on $\II$. If $\II=2^{I}$ then $R$ is sampled according to $\bx$, i.e., $R\sim (1-\delta)^2\bx$ (as defined in Section~\ref{sec:related}). If $\II$ is a matroid constraint, the sampling of \cite{CVZ10} is used. Finally, if $\II$ is a matroid intersection, or a matching constraint, then the dependent rounding technique of \cite{CVZ11} is used. Each of the distributions admits a Chernoff-like concentration bound. These bounds are central to our proof of correctness. We refer to the above operation as sampling $R$ by $\bx$, $\delta$ and $\II$. 

\cout{
While each sampling method is slightly different, we use them obliviously via the following lemma.
\begin{lemma}
	\label{lem:sampling}
	Let $(\II ,f)$ be a valid pair,  $\bx \in P(\II)$ and $\delta>0$. Then, there is a polynomial time random algorithm which, given $\II$, $\bx$ and $\delta$, returns a random set $R\in \II$ for which the following hold.
	\begin{enumerate}
		\item For any $i\in I$, $\Pr(i\in R) = (1-\delta)^2 \bx_i$.
		\item For any $\ba\in [0,1]^I$,  $\eps>0$  and $\zeta \geq (1-\delta)^2 \ba \cdot \bx$,
		$$\Pr\left(\sum_{i\in R} \ba_i \geq (1+\eps) \zeta \right)\leq \exp\left( -\frac{\zeta\cdot \delta\cdot \eps^2}{20}\right).$$
		\item Let $\eps>0$, $f_{\max}= \max_{i\in I} f(\{i\})-f(\emptyset)$ and $\zeta \leq (1-\delta)^2 F(\bx)$, where $F$ is the multilinear extension of $f$. Also, assume $f$ is monotone.  
		Then,
		$$\Pr\left(f(R) \leq (1-\eps)\zeta \right)\leq \exp\left(- \frac{\zeta\cdot  \delta\cdot \eps^2}{20 \cdot f_{\max}}\right).$$
		\item If $\II=2^I$ then the events $(i\in R)_{i\in I}$ are independent.
	\end{enumerate}
	We say that the set $R$ is sampled by $\bx$, $\delta$ and $\II$.
\end{lemma}
Lemma \ref{lem:sampling} follows directly from \cite{CVZ10}, \cite{CVZ11}, the concentration bound of \cite{Vo10}  and the Chernoff bounds of \cite{RKPS06}. The requirement for $(\II,f)$ to be valid stems for the lack of concentration bounds for non-monotone and submodular functions in \cite{CVZ10} and \cite{CVZ11} respectively.  
}

Given the set $R$, the algorithm proceeds to a {\em purging} step.
 While this step does not affect the content of $R$ if $f$ is monotone, it is critical in the non-monotone case. Given a submodular function $f:2^I\rightarrow \mathbb{R}$, we define a purging function $\eta_f:2^I\rightarrow 2^I$ as follows.
 Fix an arbitrary order over $I$ (which is independent of $S$), initialize $J=\emptyset$ and iterate over the items in $S$ by their order in $I$. For an item $i\in S$, if $f(J\cup \{i\})-f(J)\geq 0$ then $J\leftarrow J\cup \{i\}$;
else, continue to the next item. Now, $\eta_f(S)=J$, where $J$ is the set at the end of the process. The purging function was introduced in \cite{CVZ14} and is used here similarly in conjunction with the FKG inequality.

While the above sampling and purging steps can be used to select a set of items for the solution, they do not determine how these items are assigned to the bins. We now show that it suffices to associate the selected items with blocks and then use a Bin Packing algorithm for finding their assignment to the bins in the blocks, as in Algorithm~\ref{alg:uniform}. 

Intuitively, we would like to associate a subset of items $I^t_j$ with a block $K^t_j$ in a way
that enables to assign the items in $I^t_j\cap R$ to $|K^t_j|$ bins, for $1\leq t\leq d$ and $1 \leq j \leq \ell_t$. Consider two cases. If $|K^t_j| >1$
 then we ensure $I^t_j\cap R$ satisfies conditions that allow using Fractional Grouping (see Lemma~\ref{lem:grouping}). On the other hand, if $|K^t_j|=1$, it suffices to require that $R\cap I^t_j$ 
 adheres to the capacity constraint of this bin. Such a partition $(I^t_j)_{j=0}^{\ell_t}$ of $\supp(\bx)$ can be computed for each of the MKCs. 
We refer to this partition as the   {\em Block Association} of a point in the $\gamma$-partition polytope and $\aeps$, on which the partition depends.
 In Appendix~\ref{sec:association} we give a formal definition of block association and its  properties.
 
 We proceed to analyze our algorithm  (see the pseudocode in Algorithm~\ref{alg:restricted}). 

\begin{algorithm}
	\SetAlgoLined
	\SetKwInOut{Input}{Input}\SetKwInOut{Configuration}{Configuration}
	\DontPrintSemicolon
	
	\KwIn{An $(N,\xi)$-restricted  $d$-MKCP instance $\mathcal{R}$ defined by  $\left(I,\left(w_t,B_t,W_t\right)_{t=1}^d,\II, f\right)$ and $(K^t_j)_{j=0}^{\ell_t}$, the $N$-leveled partition of $B_t$ for $1\leq t\leq d$.}
	\Configuration{$\gamma>0$, $\delta>0$, $N\in\mathbb{N}$, $\xi\in \mathbb{N}$,}

	Optimize $F(\bx)$ with $\bx \in P$, where $P$ is the
	$\gamma$-instance polytope of $\mathcal{R}$, and $F$ is the multilinear extension of $f$. \label{restricted:opt} \;
	
	Let $R$  be a random set sampled by $\bx$, $\delta $ and $\II$. Define $J=\eta_f(R)$ ($\eta_f$ is the purging function). 
	\label{restricted:sampling}
	\;
	
	Let $\by^{t,0},\ldots, \by^{t,\ell_t}\in [0,1]^{I}$ such that 
	$(\bx, \by^{t,0},\ldots, \by^{t,\ell_t})\in P_t^e$, where  $P_t^e$ is the extended $\gamma$-partition polytope of $(w_t, B_t,W_t)$ and the partition $(K^t_j)_{j=0}^{\ell_t}$, for $1\leq t\leq d$.
	\label{restricted:defs}\;
	
	Find the block association $(I^t_j)_{j=0}^{\ell_t}$ of $(1-\delta)(\bx, \by^{t,0},\ldots, \by^{t,\ell_t})$ and $\aeps=\frac{\delta}{4}$ for $1\leq t \leq d$. \label{restricted:association}\;
	
	Pack the items of $J\cap I_j^t$ into the bins of $K^t_j$ using an algorithm for bin packing if $|K^t_j|>1$, or simply assign $J\cap I_j^t$ to  $K^t_j$ otherwise . \label{restricted:packing} \;
	
	Return $J$ and the resulting assignment if the previous step succeeded; otherwise, return an empty set and an empty packing. \label{restricted:return}
	\;
	
	\caption{$(N,\xi)$-restricted $d$-MKCP} 
	\label{alg:restricted}
\end{algorithm}

\begin{lemma}
	\label{lem:restricted_approx}
	For any $d\in \mathbb{N}$, $\eps>0$ and $M>0$, there are parameters $N\in \mathbb{N}$ satisfying $N>M$, $\xi\in \mathbb{N}$, $\gamma>0$ and $\delta>0$ such that Algorithm \ref{alg:restricted} is a randomized  $(c-\eps)$-approximation for  $(N,\xi)$-restricted  $d$-MKCP, where $c=1$ for modular instances with any type of additional constraint,  $c=1-e^{-1}$ for monotone instances with a matroid constraint, and $c=0.385$ for non-monotone instances with no additional constraint.
\end{lemma}
\cout{
\begin{proof}[Proof Sketch]
By Lemma \ref{lem:instance_soundness}, it follows that there is $S'\subseteq I$ such that $\one_{S'}\in P$ and $f(S')\geq \left(1- \frac{N^2\cdot d}{\xi \gamma}\right)\OPT \geq (1	-\eps^2)\OPT$, where the last inequality follows from a proper selection of the parameters. Thus, it holds that $F(\bx)\geq c\cdot(1-\eps^2)^2\OPT$. We note that the FPTAS of Lemma \ref{lem:instance_fptas} suffices for the algorithm of \cite{CCPV07} to yield an approximate solution for the multilinear optimization 
problem. In the modular case, Lemma~\ref{lem:instance_fptas} can be used directly to obtain $\bx$, as the multilinear extension of a modular function is linear (see Lemma \ref{lem:multilinear_of_linear}).

We use standard concentration bounds to show that 
with high probability (w.h.p)
$J\cap I^t_j$ satisfies the conditions of Fractional Grouping (Lemma~\ref{lem:grouping}); thus, w.h.p. the packing in Step~\ref{restricted:packing} succeeds.

In the monotone (and modular) case, concentration bounds are used to show that $f(J)$ is sufficiently large w.h.p. In the non-monotone case, the expected value of the returned solution is bounded using the  FKG inequality (as in \cite{CVZ14}).
\end{proof}
}
A formal proof of the lemma is given in Appendix \ref{app:algorithm}.
 Theorem \ref{thm:main} follows immediately from Lemmas \ref{lem:restricted_approx} and \ref{lem:restricted}.


\bibliographystyle{plain}
\bibliography{bibfile}

\appendix 

\section{Application of $d$-MKCP in Cloud Data Centers}
\label{app:applications}

In the Cross Region Cloud Provider problem, a cloud provider maintains two data centers, located in different geographical regions (e.g, an American and a European data center). Each of the data centers contains servers of varying computing capacities. 
Additionally, there is a collection of applications which  the provider may deploy on its cloud. Each application requires some compute capacity in each region.
As the demand for an application differs between the regions, the compute requirement of an application may differ significantly between the regions. 
Each application also has a specified profit the provider gains if the application is deployed.

The cloud provider needs to select a subset of applications to be deployed, and assign each of these applications to a server in each of the regions. The assignment of applications to servers must preserve the compute capacities of the servers. The provider's goal is to maximize the total profit obtained from the selected applications. 
The Cross Region Cloud Provider problem can be easily cast as an instance of $2$-MKCP with no additional constraint.

\section{Reduction to Restricted Instances}
\label{app:restricted}

In this section we prove Lemma \ref{lem:restricted}.
Let $\mathcal{T} = \left(I, \left( \mathcal{K}_t\right)_{t=1}^d, \II, f \right)$ be a  $d$-MKCP instance, where $\mathcal{K}_t = (w_t, B_t, W_t)$, $1\leq t\leq d$ are the $d$ multiple knapsack constraints. Also,
let $S$ and $(A_t)_{t=1}^{d}$ be some solution for $\mathcal{T}$,
and let $\xi\in \mathbb{N}$. We define the {\em residual instance} of $\mathcal{T},S,(A_t)_{t=1}^d$ and $\xi$ as the valid
$d$-MKCP instance $\mathcal{T}'=\left(I', \left( \mathcal{K}'_t\right)_{t=1}^d,  \II', g \right)$, where
\begin{enumerate}
	\item
	$I'=\left\{ i\in I\setminus S ~\middle|~ f(\{i\}\cup S )-f(S)\leq \frac{f(S)}{\xi}\right\}$.
	\item The function 
	$g:2^{I'}\rightarrow \mathbb{R}_{\geq 0}$ is defined by $g(T)=f(S\cup T)$.
	\item For any $1\leq t\leq  d$, $\mathcal{K}'_t = (w_t, B_t, W'_t)$, where $W'_t(b)= W_t(b) - w_t(A_t(b))$ for all $b\in B_t$.
	\item
	$\II' =\{ T\subseteq I'~|~T\cup S \in \II \}$. 
\end{enumerate}

We note that the function $g$ is non-negative and submodular. Furthermore (as shown in Lemma~\ref{lem:submodular_cup}), $g$ is monotone (modular) if $f$ is monotone (modular).
Also, if $\II$ describes a constraint (matroid, matroid intersection or matching constraint), $\II'$ describes a constraint of the same type. 
If $T\subseteq I'$ and its assignment $(D_t)_{t=1}^{d}$ define a solution for the residual 
instance $\mathcal{T}'$, then it can be easily verified that $S\cup T$ and $(A_t \cup D_t)_{t=1}^{d}$  form a solution for the original instance $\mathcal{T}$.\footnote{Given two assignments $A_1,A_2:B\rightarrow 2^I$, we define $D=A_1 \cup A_2$ by $D(b)=A_1(b)\cup A_b(b)$ for all $b\in B$.} Furthermore, by definition, $f(S\cup T)=g(T)$.

Residual instances are useful as the marginal value $g(\{i\})-g(\emptyset)$ of every item $i\in I'$ is bounded. 
The next lemma states that the residual instance
with a specific solution preserves the optimum.\footnote{Given an assignment $A:B\rightarrow 2^I$ and a set $S\subseteq I$, we define $A\cap S$ ($A\setminus S$) to be the assignment $D:B\rightarrow 2^I$ such that $D(b)=A(b)\cap S$ ($D(b)=A(b)\setminus S$) for every $b\in B$. }
\begin{lemma}
	\label{lem:residual_generic}
	Let $\xi\in \mathbb{N}$, $\mathcal{T}$ be an instance of  $d$-MKCP, and $S^*$,$(A^*_t)_{t=1}^d $ a solution for $\mathcal{T}$. Then there is  $S\subseteq S^*$, $|S|\leq \xi$  such that $S^*\setminus S$, $(A^*_t\setminus S)_{t=1}^d$ is a solution for the residual instance of 
	$\mathcal{T}$,  $S$, $(A^*_t\cap S)_{t=1}^{d}$ and $\xi$.
\end{lemma}
\begin{proof}
	Let  $\mathcal{T}=\left(I, \left( w_t, B_t,W_t\right)_{t=1}^d,\II, f \right)$. Also, let $S^*=\{s^*_1,\ldots, s^*_\ell\}$ with the items sorted such that
	$f(\{s^*_1,\ldots , s^*_i\})=\max_{i-1<k\leq \ell} f(\{s^*_1,\ldots ,s^*_{i-1}\}\cup \{s^*_k\})$ for every $1\leq i \leq \ell$. If $\ell\leq \xi$ let $S=S^*$, and the lemma immediately follows. Otherwise, set
	$S=\{s^*_1, \ldots, s^*_{\xi}\}$. 
	
	Let $\mathcal{T}'=\left(I', \left( w_t, B_t,W'_t\right)_{t=1}^d, \II', f \right)$ be the residual instance of 
	$\mathcal{T}$,  $S$, $(A^*_t\cap S)_{t=1}^{d}$ and $\xi$.
	By well known properties of submodular functions (see Lemma \ref{lem:marginal}),
	for every $i\in S^*\setminus S$ it holds that $f(S\cup\{i\})-f(S)\leq \frac{f(S)}{\xi}$, and therefore
	$S^*\setminus S\subseteq I'$. Furthermore, it can be easily verified that $S^*\setminus S \in \II'$ and $A^*_t\setminus S$ is a feasible assignment of $S^*\setminus S$ w.r.t $(w_t, B_t,W'_t)$; thus, $S^*\setminus S$ and $\left(A^*_t\setminus S\right)_{t=1}^{d}$ is a solution for the residual instance $\mathcal{T}'$. 
\end{proof}
The enumeration (Step \ref{reduction:loop}) of Algorithm \ref{alg:reduction} iterates over all subsets $S'\subseteq I$, $|S'|\leq \xi$ and disjoint assignments\footnote{An assignment $A:B\rightarrow 2^I$ is disjoint  if for any $b_1, b_2\in B$, $b_1\neq b_2$, it holds that $A(b_1)\cap A(b_2)=\emptyset$.} 
$(A_t)_{t=1}^{d}$ of $S'$ to the $d$ MKCs, and finds an approximate solution for each residual instance. The number of  possible subsets and assignments is polynomial for a fixed $\xi$. In one of these iterations, $S'=S$ and $A_t=A^*_t\cap S$ for every $1\leq t\leq d$ ($S$ is the set from Lemma \ref{lem:residual_generic}). In this specific iteration, the value of the optimal solution for the residual instance equals to the value of an optimal solution for the original instance; thus, the approximate solution for the residual instance can be used to derive an approximate solution for the input instance.

Algorithm~\ref{alg:reduction} converts each of the $d$ MKCs into an $N$-leveled constraint. 
The following lemma, adapted from~\cite{FKNRS20}, shows that any MKC can be converted into an $N$-leveled constraint with only a small decrease in the value of an optimal solution. 

\begin{lemma}
	\label{lem:structuring}
	For any $N$, set of bins $B$ and capacities $W:B\rightarrow \mathbb{R}_{\geq 0}$, there is $\tilde{B}\subseteq B$,  capacities $\tilde{W}:\tilde{B}\rightarrow \mathbb{R}_{\geq 0}$,
	and an $N$-leveled partition  $(\tilde{K}_j)_{j=0}^\ell$ of $\tilde{B}$,
	such that
	\begin{enumerate}
		\item $\tilde{B}$, $\tilde{W}$ and $(\tilde{K}_j)_{j=0}^\ell$ can be computed in polynomial time.
		\item
		The bin capacities satisfy $\tilde{W}(b) \leq W(b)$, for every $b\in \tilde{B}$.
		\item
		\label{structuring_prop2}
		For any set of items $I$, weight function  $w:I\rightarrow \mathbb{R}_{\geq 0}$, a subset $S\subseteq I$ feasible for the MKC $(w,B,W)$ over $I$, and a submodular function $f:2^I \rightarrow \mathbb{R}_{\geq 0}$, there is  $\tilde{S}\subseteq S$ feasible for the MKC $(w,\tilde{B},\tilde{W})$ such that $f(\tilde{S}) \geq \left(1- \frac{1}{N}\right) f(S)$.\footnote{We say $S$ is feasible for a MKC $(w,B,W)$ if there is a feasible assignment $A$ of $S$ w.r.t to $(w,B,W)$.}
	\end{enumerate}
	We refer to $\tilde{B}$ and $\tilde{W}$ as the $N$-leveled constraint of $B$ and $W$. 
\end{lemma}

The proof of the lemma requires a minor adaptation of the proof in~\cite{FKNRS20}.  
We now proceed to the analysis of Algorithm \ref{alg:reduction}.

\begin{lemma}
\label{lem:reduction_analysis}
For any $N,\xi, d \in \mathbb{N}$ and $c\in [0,1]$, if $\mathcal{A}$ is a polynomial-time random $c$-approximation for modular/ monotone/ non-monotone 
$(N,\xi)$-restricted  $d$-MKCP  with a matroid/ matroid intesection/ matching/ no additional constraint then Algorithm \ref{alg:reduction} configured with $\mathcal{A}$ is  a polynomial-time random  $c\cdot \left(1-\frac{d}{N}\right)$-approximation for $d$-MKCP with the same type of function and additional constraint.
\end{lemma}
\begin{proof}
Let $\cT=\left(I,\left(w_t,B_t,W_t\right)_{t=1}^d,\II, f\right)$ be  $d$-MKCP instance where $f$ and $\II$ matches the type of functions and additional constraint solved by $\cA$.  Also, let $S^*$ and $(A^*_t)_{t=1}^{d}$ be an optimal solution for $\mathcal{T}$, and $\OPT=f(S^*)$. W.l.o.g., we assume that $A^*_t$ is a disjoint assignment for any $1\leq t\leq d$. 

By Lemma~\ref{lem:residual_generic}, there is $S\subseteq S^*$, $|S|\leq \xi$ such that $S^*\setminus S$ and $(A^*_t \setminus S)_{t=1}^{d}$ is a solution for $\mathcal{T}'$, the residual problem of $\mathcal{T}$, $S$, $(A^*_t\cap S)_{t=1}^{d}$ and $\xi$. We focus in the analysis on the iteration of the loop in Step \ref{reduction:loop} in  which $S$ is the set defined above, and $A_t= A^*_t \cap S$ for any $1\leq t\leq d$.

It holds that $\tilde{\cT}$ is a  $(N,\xi)$-restricted  $d$-MKCP instance with the type of function and additional constraint handled by $\cA$. Its constraints are $N$-leveled due to the leveling. Furthermore, $g(\{i\})-g(\emptyset) = f(S\cup\{i\})-f(S)\leq \frac{f(S)}{\xi} \leq \frac{\OPT}{\xi}$ for any $i\in I'$  (note that $S$ and $(A^*_t\cap S)_{t=1}^{d}$ form a solution for $\cT$, thus $f(S)\leq \OPT$). Finally, it can be easily verified that the optimum of $\tilde{\cT}$ is at most $\OPT$.  We conclude that $\mathcal{A}$ returns a $c$-approximation for $\tilde{\cT}$.

By iterative application of  Lemma \ref{lem:structuring} it follows that there is $\tilde{S}\subseteq S^*\setminus S \subseteq I'$ such that  $\tilde{S}$ is feasible for $(w_t,\tilde{B}_t,\tilde{W}_t)$ for any $1\leq t\leq d$ and $g(\tilde{S})\geq \left(1-\frac{d}{N}\right) g(S^*\setminus S) = \left(1-\frac{d}{N}\right) \OPT$. Thus, there are assignments with which  $\tilde{S}$ is a feasible solution for $\tilde{\cT}$. We conclude that $\E[g(R)]\geq c\cdot \left(1-\frac{d}{N}\right) \OPT$ (the randomization stems from $\mathcal{A}$ being possibly random). 

Note that $T^*$ is maintained by the algorithm as the solution of highest value attained so far.
Thus, subsequent to the specific iteration it holds that $\E[f(T^*)]\geq \E[f(S\cup R)]= \E[g(R)]\geq \left(1-\frac{d}{N}\right) \OPT$.

We note that the algorithm has a polynomial running time, as the number of iterations is polynomial for a constant $\xi$, and that the returned solution is always feasible.

\end{proof}
Lemma \ref{lem:restricted} follows immediately from Lemma \ref{lem:reduction_analysis}.
\begin{algorithm}
	\SetAlgoLined
	\SetKwInOut{Input}{Input}\SetKwInOut{Configuration}{Configuration}
	\DontPrintSemicolon
	
	\KwIn{A $d$-MKCP instance $\cT=\left(I,\left(w_t,B_t,W_t\right)_{t=1}^d,\II, f\right)$}
	\Configuration{A $c$-approximation algorithm $\mathcal{A}$ for $(N,\xi)$-restricted $d$-MKCP}
	
	Initialize $T^*\leftarrow \emptyset$ and $D^*_t:B_t\rightarrow 2^I$ by $D^*_t(b)=\emptyset$ for every $1\leq t\leq d$ and $b\in B_t$. \;
	
	\ForAll{ $S\subseteq I$, $|S|\leq \xi$ and 
		feasible disjoint assignments $A_t$ of $S$  w.r.t. $(w_t, B_t, W_t)$ for $1\leq t\leq d$ \label{reduction:loop}
	}{
		Let $\cT'=(I', \left(w_t, B_t, W'_t\right)_{t=1}^{d}, \II', g)$ be the  residual instance of $\cT$, $S$, $(A_t)_{t=1}^{d}$ and $\xi$. \;
		
		Let $\tilde{B}_t, \tilde{W}_t$ be the $N$-leveled constraint of $B_t$, $W'_t$,  and let  $(K^t_j)_{j=0}^{\ell_t}$ be its $N$-leveled partition for $1\leq t\leq d$.
		\label{reduction:leveling}
		\;
		
		Use algorithm $\mathcal{A}$ to find approximate solution for  $\tilde{\cT}=(I' , \left(w_t, \tilde{B}_t, \tilde{W}_t\right)_{t=1}^{d}, \II', g)$  with the partitions $(K^t_j)_{j=0}^{\ell_t}$  for $1\leq t\leq d$. Let $R $ and $(A'_t)_{t=1}^{d}$ be the returned solution \; 
		
		\If{$g(R)\geq f(T^*)$}{
			Set $T^*\leftarrow S\cup R$ .\;
			
			For $1\leq t \leq d$: set $D^*_t(b) \leftarrow A'_t(b) \cup A_t(b)$  for $b\in \tilde{B}_t$ and $D^*_t(b)\leftarrow A_t(b)$ for $b\in B_t\setminus \tilde{B}_t$.\;
		}

	}
	Return $T^*$ and $\left(D^*_t\right)_{t=1}^{d}$.\;
	\caption{Reduction to Restricted $(N,\xi)$-MKCP} 
	\label{alg:reduction}
\end{algorithm}

\section{Block Association}
\label{sec:association}
In this section we present our block association technique, formally summarized in the next lemma.
\begin{lemma}[Block Association]
	Let $(w,B,W)$ be an MKC, and $(K_j)_{j=0}^{\ell}$ a partition of $B$ into blocks.
	Suppose that $(\bx , \by^0,\ldots, \by^{\ell})\in P^e$, where $P^e$ is the $\gamma$-extended partition polytope of $(w,B,W)$ and $(K_j)_{j=0}^{\ell}$ for some $\gamma>0$. Then, given a parameter $\aeps>0$, there is a polynomial-time algorithm which finds a partition $(I_j)_{j=0}^{\ell}$ of $\supp(\bx)$
	satisfying the following conditions for all $0\leq j \leq \ell$.\footnote{We denote $\supp(\bnu)=\{i\in I~|~ \bnu_i >0\}$ for $\bnu\in [0,1]^{I}$.}
	\begin{enumerate} 
		\label{lem:association}
		\item
		If  $|K_j|>1$, let $G^j_1, \ldots ,G^j_{\tau_j}$ be the $\aeps$-grouping of~$\by^j$. Then
		$\sum_{i\in I_j \cap G^j_k } \bx_i \leq \aeps |K_j| +2$ for all $1\leq k \leq \tau_j$, and there is $i^*_j\in I$ such that  $\sum_{i\in I_j \cap L_{K_j,\aeps}\setminus \{i^*_j\} } \bx_i w(i) \leq \sum_{i\in  L_{K_j,\aeps} } \by^j_i  w(i)$.
		\item 
		If  $|K_j|=1$ 
		there is $i^*_j\in I$ such that $\sum_{i\in   I_{j}\setminus \{i^*_j\}} \bx_i \cdot w(i)\leq \sum_{i\in I } \by^j_i\cdot  w(i)$. 
		\item
		It holds that 
		$I_j\subseteq \supp(\by^j)$.
	\end{enumerate}
	We refer to $(I_j)_{j=0}^{\ell}$ as the {\em Block Association} of $(\bx , \by^0,\ldots, \by^{\ell})\in P^e$.
\end{lemma}

Defining a block association for the items relies on an abstract notion of constraints. A {\em constraint} is a pair  $(\bc, \beta)$ of coefficients  $\bc\in \mathbb{R}_{\geq 0}^{I}$ and a bound 
$\beta \in \mathbb{R}_{\geq 0}$. 
We say that $\bgam \in \mathbb{R}^{I}$ {\em satisfies} the 
constraint $(\bc,\beta)$ if $\sum_{i\in I} \bc_i \cdot \bgam_i \leq \beta$, and that $\bgam$ {\em semi-satisfies} $(\bc, \beta)$ if $\exists i^*\in I$ such that $\sum_{i\in I\setminus \{i^*\}} \bc_i\cdot  \bgam_i \leq \beta$.
We say that $(\bgam^r)_{r=1}^{p}$, $\bgam^r \in [0,1]^{I}$  are a {\em decomposition} of $\bx\in [0,1]^{I}$ if $\bx= \sum_{r=1}^{p} \bgam^r$. Item $i\in I$ is {\em perfect} w.r.t. vectors $(\bgam^r)_{r=1}^{p}$ if there is at most one vector $\bgam^r$ for which $\bgam^r_i\neq 0$; otherwise, $i$ is {\em broken}. We say that a vector $\bgam^r\in [0,1]^I$ is {\em perfect}  w.r.t $(\bgam^r)_{r=1}^{p}$ if all items in $\supp(\bgam^r)$ are perfect (recall $\supp(\bnu)=\{i\in I~|~ \bnu_i >0\}$ for $\bnu\in [0,1]^{I}$). Otherwise, $\bgam^r$ is {\em broken}. 
A decomposition $(\bgam^r)_{r=1}^{p}$ is {\em perfect} if all items are perfect (and therefore all vectors $\bgam^r$ are perfect). 
Given vectors $(\blam^r)_{r=1}^{p}$, we define the {\em broken bipartite graph of $(\blam^r)_{r=1}^{p}$ } as the bipartite graph $G=(S,T,E)$ where $S$ is the set of broken items, $T=\{ 1\leq r \leq p~|~\text{$\blam^r$ is broken}\}$ is the set of broken vectors, and $E=\{(i,r)~|~ \blam^r_i\neq 0 \}$ (see Figure~\ref{fig:broken_graph_example}).

\begin{figure}
	\centering
	\begin{subfigure}{.5\textwidth}
		\centering
		\begin{tikzpicture}
		\node[] at (1,3.25) {S};
		\node[] at (4,3.25) {T};
		\node[default node](1) at(1,2.25) {$1$};
		\node[default node](2) at(1,0.75) {$3$};
		\node[default node](3) at(4,2.25) {$1$};
		\node[default node](4) at(4,0.75) {$3$};
		
		\draw  (1)--(3);
		\draw  (2)--(3);
		\draw  (2)--(4);
		\draw  (1)--(4);
		\end{tikzpicture}
		\caption{Broken graph $G=(S,T,E)$ with $S=\{1,3\}$, $T= \{1,3\}$, and $E=\{(1,1),(1,3),(3,1),(3,3)\}$ for the following instance $I=\{1,2,3\}$, $p=3$,
			$\bgam^{1}=(\frac{1}{3},0,\frac{2}{3})$, 
			 $\bgam^{2}=(0,1,0)$, and $\bgam^{3}=(\frac{2}{3},0,\frac{1}{3})$.}
		\label{fig:broken_graph_example}
	\end{subfigure}%
	\begin{subfigure}{.5\textwidth}
		\centering
		\begin{tikzpicture}
		\node[] at (9,4) {S};
		\node[] at (12,4) {T};
		\node[default node](3) at(9,0) {$3$};
		\node[default node](2) at(9,1.5) {$2$};
		\node[default node](1) at(9,3) {$1$};
		\node[default node](6) at(12,0) {$3$};
		\node[default node](5) at(12,1.5) {$2$};
		\node[default node](4) at(12,3) {$1$};
		
		\draw  (2)--(5);
		\draw  (2)--(4);
		\draw  (3)--(5);
		\draw [ultra thick] (1)--(4) node[pos=0.3,above,label]{$+\bnu_1$};
		\draw [ultra thick] (3)--(4) node[pos=0.2,above,label,sloped]{$-\bnu_2$};
		\draw [ultra thick] (3)--(6) node[pos=0.3,below,label]{$+\bnu_2$};
		\draw [ultra thick] (1)--(6) node[pos=0.2,above,label,sloped]{$-\bnu_1$};
		\end{tikzpicture}
		\caption{Shift items along cycle $C=(i_1,r_1,i_2,r_2)$, with $i_1=r_1=1$ and $i_2=r_2=3$.}
		\label{fig:simple_cycle_example}
	\end{subfigure}
	\label{fig:make_perfect}
	\caption{Visualizations of procedures in Algorithm \ref{alg:perfect-decomp}}
\end{figure}
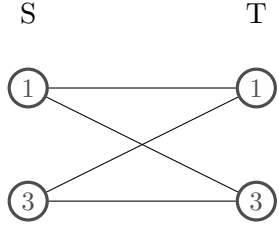
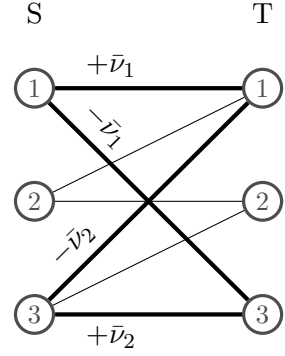

\begin{algorithm}
	\SetAlgoLined
	\SetKwInOut{Input}{Input}\SetKwInOut{Configuration}{Configuration}
	
	\DontPrintSemicolon
		\KwIn{Vectors $\bx,(\bgam^r)_{r=1}^{p}$, constraints $(\bc^r,\beta^r)$ for $1\leq r\leq p$.}
	
	Initialize $\blam^r=\bgam^r$ for $1\leq r\leq p$\; 
	
	\While {there exists a broken item $i$ (w.r.t. $(\blam^r)_{r=1}^{p}$)}{
		Update  $G=(S,T,E)$ to be  the  broken bipartite graph of $(\blam^r)_{r=1}^{p}$.\;
		\If {there exists a broken vector $r\in T$ such that $\deg_G(r)=1$\label{line:if-statement}} {
			Let $i\in S$ be the neighbor of $r$. Set $\blam^r_i\leftarrow \bx_i$, and for all $r'\neq r$ set $\blam^{r'}_i\leftarrow 0$.\label{line:if-update}
		}
		\Else {
			Let $C=(i_1, r_1,  \ldots,i_k,r_k)$ be a simple cycle in $G$. \label{perfect:cycle}\;
			Find a vector $\bnu\in\mathbb{R}^k, \bnu\neq 0$ such that $\bc^{r_{j}}_{i_j} \cdot \bnu_j - \bc^{r_{j}}_{i_{j+1}} \cdot \bnu_{j+1} \leq 0$ for $1\leq j \leq k$ ($\bnu_{k+1}=\bnu_1$, $i_{k+1}=i_1$). \label{line:shifting-vector-cond}\label{perfect:nu}\; 
			Select a maximal $\delta$ such that $\blam^{r_j}_{i_j} +\delta \bnu_j, \blam^{r_{j}}_{i_{j+1}} -\delta \bnu_{j+1}\geq 0$ for every $1\leq j \leq k$.\label{line:select-max-delta}\; 
			Set $\blam^{r_j}_{i_j} \leftarrow  \blam^{r_j}_{i_j} +\delta \bnu_j$ and $\blam^{r_{j}}_{i_{j+1}} \leftarrow \blam^{r_{j}}_{i_{j+1}} -\delta \bnu_{j+1} $ for $1\leq j \leq k$ .\label{line:else-update}\; 
		}
	
	}
	Return 	$(\blam^r)_{r=1}^{p}$
	\caption{Make Perfect} 
	\label{alg:perfect-decomp}
\end{algorithm}

\begin{lemma}
	\label{lem:abstract_association}
Given a set of constraints $\{(\bc^r, \beta^r)~|~1\leq r\leq p\}$ and a decomposition $(\bgam^r)_{r=1}^{p}$ of $\bx\in [0,1]^I$, such that $\bgam^r$ satisfies $(\bc^r,\beta^r)$ for $1\leq r\leq p$, Algorithm \ref{alg:perfect-decomp}
	returns in polynomial time a  perfect decomposition  $(\blam^r)_{r=1}^{p}$ of $\bx$ such that $\blam^r$ semi-satisfies $(\bc^r,\beta^r)$, and $\supp(\blam^r)\subseteq \supp(\bgam^r)$ for $1\leq r\leq p$. 
\end{lemma}

 We first show that the algorithm is well defined, as stated in the next lemma.
\begin{lemma}\label{lem:well-defined}
	Algorithm \ref{alg:perfect-decomp} is well defined. That is, if the statement in Step \ref{line:if-statement} is false, then there  exists a cycle and a vector $\bnu$ as defined in  Steps \ref{perfect:cycle}, \ref{perfect:nu}.
\end{lemma}
\begin{proof} 
	
	We note that all of the items in the broken graph (the set $S$) are of degree two or more by definition, and if Step \ref{perfect:cycle} is reached then all the vectors $r\in T$ are of degree at least two. It follows that the broken graph $G$ contains a cycle.
	
	The existence of $\bnu$ in Step \ref{perfect:nu} follows from a simple argument in linear algebra. 
	The equations  $\bc^{r_{j}}_{i_j} \cdot \bnu'_j - \bc^{r_{j}}_{i_{j+1}} \cdot \bnu'_{j+1} = 0$ for  $1\leq j \leq k-1$ form $k-1$ linear homogeneous  equations with $k$ variables.
	Thus, there  is $\bnu'\in \mathbb{R}^k$, $\bnu'\neq 0$ which satisfies the equations. Hence, either $\bnu =\bnu'$ or $\bnu= -\bnu'$ satisfies the condition in Step \ref{perfect:nu} (i.e., the $k$-th inequality holds as well).
\end{proof}

To prove Lemma \ref{lem:abstract_association}, we show that the number of edges in $G$ decreases in each iteration, and that the following invariants hold throughout the execution of the algorithm.
\begin{enumerate}
	\item 
	For any $1\leq r\leq p$, either   $\blam^r$ is perfect and  semi-satisfies $( \bc^r, \beta^r)$, or it satisfies $(\bc^r, \beta^r)$.
	\item
	For any $1\leq r\leq p$ it holds that  $\supp(\blam^r)\subseteq \supp(\bgam^r)$. 
	\item
	$\left(\blam^r\right)_{r=1}^{p}$ is a decomposition of $\bx$. 
\end{enumerate}
These remaining parts are fairly simple, hence deferred to Appendix \ref{app:association}.


\begin{proof}[Proof of Lemma \ref{lem:association} (Block Association)]

	Let $M = \left\{ 0\leq j\le \ell  ~|~|K_j|>1 \right\}$ be the blocks which contain multiple bins. Also, we use  $\wedge$ as an element-wise minimum over two vectors. That is, for $\bnu^1, \bnu^2\in \mathbb{R}^I$ we define  $\bnu^1 \wedge \bnu^2 =\bgam$, where  $\bgam\in \mathbb{R}^I$ and $\bgam_i = \min\{ \bnu^1_i, \bnu^2_i\}$ for all $i\in I$.

	Let $R=\{(j,k) ~|~j\in M,~1\leq k \leq \tau_j \}\cup \{j~|~0\leq j\leq \ell\}$.
	We define a set of constraints $\{(\bc^r,\beta^r)~|~r\in R\}$ for $(\bx,\by^0,\ldots, \by^{\ell})$.  
	For every $i\in I$ we define $\bc^j_i=w(i)$ for any $0\leq j\leq \ell$ and $\bc^{j,k}_i=1$ for any $j\in M$ and $1\leq k\leq \tau_j$. Also, we define $S^j = L_{K_j,\aeps}$ for $j\in M$, $S^{j,k}=G^j_k$ for $j\in M$ and $1\leq k\leq \tau_j$, and $S^j =\{i\in I~|~w(i)\leq W^*_K\}$ for $j\notin M$. 
	We define $\bgam^r=\by^j \wedge \one_{S^r}$ for $r=(j,k)$ or $r=j$, and $\beta^r = \bc^r\cdot \bgam^r$. By definition, it holds that $\bgam^r$ satisfies $(\bc^r, \beta^r)$. 
	Furthermore, it can be easily verified that $(\bgam^r)_{r\in R}$ is a decomposition of $\bx$.
	
	By Lemma~\ref{lem:abstract_association},  given $\bx$, $(\bgam^r)_{r\in R}$  and the constraints $(\bc^r, \beta^r)_{r\in R}$  Algorithm \ref{alg:perfect-decomp}
	 returns a perfect decomposition 
	$(\blam^r)_{r\in R}$ of $\bx$ such that $\blam^r$ semi-satisfies $(\bc^r,\beta^r)$. Define $I_j = \supp(\blam^j) \cup \left(\bigcup_{k=1}^{\tau_j}\supp(\blam^{j,k})\right)$ for $j\in M$, and $I_j=\supp(\blam^j)$ for $0\leq j\leq \ell$, $j\notin M$.  	
	It follows from Lemma \ref{lem:abstract_association} that $(I_j)_{j=0}^{\ell}$ can be computed in polynomial time. 
	
	For every $r\in R$ it holds that $\supp(\blam^r)\subseteq \supp(\bgam^r)\subseteq S^r$. Furthermore, it can be easily verified that for any $r\in R$, $r=j$ or $r=(j,k)$, it holds that $I_j \cap S^r = \supp(\blam^r)$. By the same argument, it follows that
	$I_j\subseteq \supp(\by^j)$. 
	 Finally, we note that $(I_j)_{j=0}^{\ell}$ is a partition of $\supp(\bx)$,  since$(\blam^r)_{r\in R}$ is a perfect decomposition of $\bx$. 
	
	
	For every $j\in M$ and $1\leq k\leq \tau_j$, it holds that $\blam^{j,k}$ semi-satisfies $(\bc^{j,k},\beta^{j,k})$. Hence, there is $i^*\in I$, such that
	$$
	\sum_{i\in I_j \cap G^j_k} \bx_i = \sum_{i\in \supp(\blam^{j,k})} \bx_i =\sum_{i\in I} \blam^{j,k}_i \leq
	1+\sum_{i\in I\setminus \{i^*\}} \blam^{j,k}_i \leq  1+\beta^{j,k} \leq 1+\sum_{i\in G^j_k} \by^j_i \leq 2 + \aeps\cdot |K_j|.
	$$
	The second equality holds since $(\blam^r)_{r\in R}$ is a perfect  decomposition, and the last inequality follows from Lemma \ref{lem:grouping}. Similarly, $\blam^j$ semi-satisfies $(\bc^j, \beta^j) $. Hence, there is $i_j^*$ such that 
	\begin{equation*}
	\begin{aligned}
	\sum_{i\in I_j \cap L_{K_j,\aeps}\setminus\{i^*_j\}} & \bx_i \cdot w(i) =
	\sum_{i\in I_j \cap S^j\setminus\{i^*_j\}} \bx_i \cdot w(i)= 
	\sum_{i\in \supp(\blam^j) \setminus\{i^*_j\}} \bx_i \cdot \bc^j_i
	= \sum_{i\in I
		\setminus\{i^*_j\}} \blam^j_i \cdot \bc^j_i\\ 
	&\leq \beta^j
	= \sum_{i\in I} w(i)\cdot \bgam^j 
	=  \sum_{i\in I} w(i)\cdot \left(\by^j\wedge \one_{L_{K_j,\aeps}}\right)_i = 
	\sum_{i\in L_{K_j,\aeps} } w(i) \cdot \by^j_i.
	\end{aligned}
	\end{equation*}
	
	Finally, for $0\leq j \leq \ell$, $j\notin M$ it holds that $(\blam^j)$ semi-satisfies $(\bc^j,\beta^j)$. Thus, there is $i^*_j\in I$ such that,
	$$ 
	\sum_{i\in I_j \setminus \{i^*_j\}} \bx_i \cdot w(i) =\sum_{i\in \supp(\blam^j) \setminus \{i^*_j\}} \bx_i \cdot w(i)  = \sum_{i\in I \setminus \{i^*_j\}} \blam^j_i \cdot \bc^j_i \leq \beta^j = \sum_{i\in I} w(i)\cdot \by^j_i.
	$$
\end{proof}

\subsection{Block Association: Deferred Proofs}
\label{app:association}

\begin{lemma}\label{lem:broken-graph-progress}
After each iteration of Algorithm \ref{alg:perfect-decomp} the number of edges in the broken bipartite graph $G$ decreases by at least one.
\end{lemma}

\begin{proof}
It is easy to see that if the statement in Step \ref{line:if-statement} is true, one broken item becomes perfect and its edges are removed. Otherwise, as proven in Lemma \ref{lem:well-defined} a non-trivial vector $\bnu$ exists for which the conditions in Step \ref{line:shifting-vector-cond} hold. In Step \ref{line:select-max-delta} the maximal $\delta$ is selected such that $\blam^r_i=0$ for at least one item-vector pair $(i,r)$. This means that $(i,r)$ is removed from $G$.
\end{proof}

\begin{lemma}\label{lem:semi-satify}
Given a decomposition $(\bgam^r)_{r=1}^{p}$ of $\bx$ and constraints $(c^r,\beta^r)$ for $1\leq r\leq p$ such that $\bgam^r$ satisfies $(c^r,\beta^r)$, through out its run Algorithm \ref{alg:perfect-decomp} maintains a set of vectors $(\blam^r)_{r=1}^{p}$ such the following holds.
\begin{enumerate}
	\item 
	For any $1\leq r\leq p$, either   $\blam^r$ is perfect and  semi-satisfies $( \bc^r, \beta^r)$, or it satisfies $(\bc^r, \beta^r)$.
	\item
	For any $1\leq r\leq p$ it holds that  $\supp(\blam^r)\subseteq \supp(\bgam^r)$. 
	\item
	$\left(\blam^r\right)_{r=1}^{p}$ is a decomposition of $\bx$. 
\end{enumerate}
\end{lemma}

\begin{proof}
For each vector $\blam^r$ let $\blam^r_j$ be the value of $\blam^r$ after the $j$-th iteration of the algorithm ($(\blam^r_0)_{r=1}^p$ are the values set upon initialization). It is easy to see that all three requirements hold for $(\blam^r_0)_{r=1}^p$ as $\blam^r_0=\bgam^r$ for $1\leq r\leq p$, and the requirements hold for $(\bgam^r)_{r=1}^{p}$.

We will prove that the requirements hold after iteration $j+1$ under the assumption that all requirements hold after the $j$-th iteration, i.e., $(\blam^r_j)_{r=1}^p$ is a decomposition of $\bx$, each vector $\blam^r_j$ is either perfect and semi-satisfies $(\bc^r,\beta^r)$ or it satisfies $(\bc^r,\beta^r)$, and $\supp(\blam^r_j)\subseteq\supp(\bgam^r)$. If a vector $\blam^r_j$ is perfect, no changes are made to it, i.e., $\blam^r_{j+1}=\blam^r_j$. As it did not change, its support does not change as well and it still semi-satisfies $(\bc^r,\beta^r)$. 

Next, we will consider the case where it isn't perfect. This means that $\blam^r_j$ it satisfies $(\bc^r,\beta^r)$. Changes to vector $\blam^r_j$ can only be made in Steps \ref{line:if-update} or \ref{line:else-update}. If it changed in Step \ref{line:if-update}, it becomes perfect and its support is unchanged. 
Let $i$ be the item of Step \ref{line:if-update},  then
$$\sum_{i'\in I\setminus \{i\}} (\blam^r_{j+1})_{i'} \cdot \bc^r_{i'}= 
\sum_{i'\in I\setminus \{i\}} (\blam^r_{j})_{i'} \cdot \bc^r_{i'}  \leq \sum_{i'\in I} (\blam^r_{j})_{i'} \cdot \bc^r_{i'} \leq \beta^r,$$ 
where the last inequality holds since $\blam^r_j$ satisfies $(\bc^r,\beta^r)$. 
That is,  $\blam^r_{j+1}$ is perfect and semi-satisfies $(\bc^r,\beta^r)$.

Otherwise, changes were made to vector $\blam^r_j$ in Step \ref{line:else-update}. As items are shifted along positive edges in cycle $C$, no new items are added to the support of $\blam^r_j$, i.e., $\supp(\blam^r_{j+1})\subseteq\supp( \blam^r_j)\subseteq \supp(\bgam^r)$. In addition, since $\blam^r_j$ satisfies $(\bc^r,\beta^r)$ and due to the condition in Step \ref{line:shifting-vector-cond} it holds that $\bc^r\cdot\blam^r_{j+1}\leq\bc^r\blam^r_j\leq\beta^r$. Thus $\blam^r_{j+1}$ satisfies $(\bc^r,\beta^r)$.

Finally, we note that $(\blam^r_{j+1})_{r=1}^p$ remains a decomposition of $\bx$ as any changes made in Steps \ref{line:if-update} and \ref{line:else-update} shift items between vectors and do not change their total sum (see illustration in Figure \ref{fig:simple_cycle_example}).
\end{proof}

\section{Solving the Instance Polytope}
\label{app:poly_fptas}

In this section we prove Lemmas \ref{lem:instance_fptas}.
The proof utilizes the ellipsoid method with  separation oracles
via the result of Gr{\"{o}}tschel, Lov{\'{a}}sz and Schrijver \cite{GLS81} (see \cite{GLS93} for a comprehensive survey). As we cannot provide an exact separation oracle, or even a weak oracle (as defined in \cite{GLS81}), we follow a known scheme, which appeared for example in \cite{KK82,FGMS11}, in which the oracle may fail, and the ellipsoid method will be aborted. In this case, however,  the input
for the separation oracle can be used to derive an approximate solution. 
As in \cite{FGMS11} and \cite{KK82} we use the dual program to replace the exponential number of variables with an exponential number of constraints (for which we provide a separation oracle).  However, as the proof of Lemma \ref{lem:instance_fptas} involves matroid and matching polytopes, in a naive representation of the polytope both primal and dual programs have an exponential number of constraints and variables. To circumvent this issue a variant of the technique is used.

We start by providing an approximate separation oracle for the block polytope.
\begin{lemma}
\label{lem:sep_primal}
There is an algorithm which given an MKC $(w,B,W)$, a block $K\subseteq I$, $\by\in \mathbb{R}^{I}$ and $\eps>0$ either determines that $(1-\eps)\cdot\by \in  P_K$ or finds $\bnu \in [0,1]^I$ such that $\by \cdot \bnu > \by' \cdot \bnu$  for every $\by'\in P_K$, where $P_K$ is the block polytope of $K$.
\end{lemma}
\begin{proof}
Let $(w,B,W)$, $K$, $\by$ and $\eps$ be as defined in the lemma. We first consider the following linear program, in which the variables are $\bz_C$  for $C\in \cC_K$ ($\cC_K$ are the $K$-configurations, as defined in Section \ref{sec:grouping}).
\begin{equation}
\label{eq:block_primal}
\begin{aligned}
&\min &&~~~~~ \sum_{C\in \cC_K} \bz_C\\
&\textnormal{such that} &
\forall i\in I:& ~~~~~\sum_{C\in \cC_K\textnormal{ s.t. } i\in C} \bz_C \geq \by_i\\
&&\forall C\in \cC_K:& ~~~~~\bz_C\geq 0
\end{aligned}
\end{equation}
By Definition \ref{def:block_polytope}, it holds that $\by\in P_K$  if and only if the optimal solution for \eqref{eq:block_primal} is $|K|$ or less. The dual of the above linear program is the following. 
\begin{equation}
\label{eq:block_dual}
\begin{aligned}
&\max &&~~~~~ \sum_{i\in I} \bbeta_i \cdot \by_i \\
&\textnormal{such that} &
\forall C\in \cC_K:& ~~~~~\sum_{i\in C} \bbeta_i \leq 1\\
&&\forall i\in I:& ~~~~~\bbeta_i\geq 0
\end{aligned}
\end{equation}
We note that the first constraint in \eqref{eq:block_dual} is essentially an instance of the knapsack problem, for which there is a known FTPAS \cite{IK75} (also appears in textbooks such as \cite{Va13}). 
Let $\bbeta\in \mathbb{R}^I_{\geq 0}$ be a vector in the feasible region on \eqref{eq:block_dual} and let $\by' \in P_K$. Thus there is $\bz'\in [0,1]^{\cC_K}$ such that $(\by',\bz')\in P^e_K$ ($P^e_K$ is the extended block polytope of $K$). It holds that
\begin{equation}
\label{eq:y_beta_bound}
\bbeta \cdot \by' = \sum_{i\in I} \bbeta_i \cdot \by'_i \leq \sum_{i\in I}\bbeta_i \sum_{C\in \cC_k \textnormal{ s.t. } i\in C} \bz'_C = \sum_{C\in \cC_K}\bz'_C \sum_{i \in C} \bbeta_i \leq \sum_{C\in \cC_K} \bz'_C \leq |K|.
\end{equation}
The first and last inequalities follow from the definition of the extended block polytope (Definition \ref{def:extended_block}), the second inequality holds since $\bbeta$ is in the feasible region of \eqref{eq:block_dual}. In the following we will attempt to find a vector $\bbeta$ for which $\bbeta \cdot \by > |K|$.

For any $v\geq 0$ define,
\begin{equation}
\label{eq:dual_bound}
D_v = \left\{ \bbeta \in \mathbb{R}^I_{\geq 0} ~\middle|~ \begin{array}{ll}
	& \bbeta \cdot \by \geq v \\
	\forall C\in \cC_K:& \sum_{i\in C}\bbeta_i\leq 1
	\end{array}\right\}.
\end{equation}
Clearly, if $D_v\neq \emptyset$ then the optimal value of \eqref{eq:block_dual} is at least $v$, and otherwise the value is less than $v$. 
Given $v\geq 0$ and $\delta>0$ 
we can use the the ellipsoid method along with a separation oracle to determine if $D_v$ is empty or not. Given $\bbeta\in \mathbb{R}_{\geq 0}^{I}$ the separation oracle will verify that $\bbeta \cdot \by \geq v$, and otherwise it will return $\bbeta \cdot\by\geq v$ as the separating hyper-plane. If $\bbeta \cdot \by \geq v$, the separation oracle executes an FPTAS for the knapsack problem (e.g., \cite{IK75}) with the following instance and $\eps$. The set of items is $I$, with weight $w(i)$ and profit $\bbeta_i$ for item $i\in I$, and knapsack (bin) capacity $W^*_K$. If the approximation scheme finds a set $C$ of items
such that $\sum_{i\in C} \bbeta_i >1$, then $C\in \cC_K$ (as this is the capacity constraint of the knapsack). Hence, the separation oracle can 
return $\sum_{i\in C} \bbeta_i \leq 1$ as the violated constraint. If the set $C$ returned by the approximation scheme satisfies $\sum_{i\in C} \bbeta_i\leq 1$ then the separation oracle aborts the ellipsoid algorithm and returns $\bbeta$ .

The execution of the ellipsoid given $v$ can end in  two possible ways.  It will either determine that $D_v=\emptyset$, or the separation oracle will abort it. If the execution is aborted, let $\bbeta$ be the returned vector. Then 
$\bbeta\cdot \by \geq v$ and for $C\in \cC_K$ it holds that $\sum_{i\in C} \bbeta_i \leq \frac{1}{1-\eps}$ (otherwise the optimal solution for the knapsack has value greater than $\frac{1}{1-\eps}$ and the approximation scheme will find a solution of value greater than $(1-\eps)\frac{1}{1-\eps}=1$). Hence $(1-\eps)\bbeta\in D_{(1-\eps)v}$. Thus, in time polynomial in the instance and $\eps^{-1}$ the algorithm either finds $(1-\delta)\bbeta\in  D_{(1-\delta)v}$ or asserts $D_v=\emptyset$. 

Thus, using a binary search   over the values of $v$ we can find $(1-\eps)v_2<v_1<v_2$ and $\bbeta\in [0,1]^{I}$ such that $D_{v_2}=\emptyset$ and $\bbeta \in D_{v_1}$, in time polynomial in the input size and $\eps^{-1}$. 
 If $v_1>|K|$  then $\bbeta\cdot \by >|K|$, and the algorithm can return $\bnu = \bbeta$ due to \eqref{eq:y_beta_bound}. Otherwise, $v_2\leq \frac{|K|}{1-\eps}$.
 Therefore, the optimal solution for \eqref{eq:block_dual} is at most $\frac{|K|}{1-\eps}$, and by duality, the optimal solution for \eqref{eq:block_primal} is at most $\frac{|K|}{1-\eps}$, thus $(1-\eps)\by \in P_K$.  The overall running time of the algorithm is polynomial in the input size and $\eps^{-1}$. 
\end{proof}

\begin{proof}[Proof of Lemma \ref{lem:instance_fptas}]
Let $\left(I,\left(w_t,B_t,W_t\right)_{t=1}^d,\II, f\right)$ be a $d$-MKCP instance,  $\left(K^t_j\right)_{j=1}^{\ell_t}$ be a partition  of $B_t$  to block for $1\leq t\leq d$ , $\bc\in\mathbb{R}^I$ and $\gamma>0$. Let
$P_t$ be the $\gamma$-partition polytope of $(w_t, B_t,W_t)$ and $(K_t)_{j=1}^{\ell}$ for $1\leq t\leq d$ and $P=P(\II) \cap \left( \bigcap_{t=1}^d P_t\right)$ be the $\gamma$-instance polytope of the instance. Finally, let $\OPT= \max_{\bx\in P} \bc \cdot \bx$.

As in the proof of Lemma \ref{lem:sep_primal} we will use a binary search and an approximate separation oracle, this time to an extended form of $P$. Let $P_{K^t_j}$ be the block polytope of $K^t_j$ in the MKC $(w_t, B_t, W_t)$ for $1\leq j \leq \ell_t$ and $1\leq t\leq d$. Given $v\geq 0$, consider the polytope $D_v$, such that each element in $D_v$ is of the form $\left(\bx, \left(\left( \by^{t,j}\right)_{j=1}^{l_t}\right)_{t=1}^{d}\right)$, i.e, 
$D_v\subseteq [0,1]^I \times  \left(\times_{t=1}^{d}\times_{j=1}^{\ell_t} [0,1]^I\right)$ and contains  all the vectors which satisfy the following constraints. 
\begin{align}
&&~~~&\bx\cdot \bc\geq v \label{eq:fptas_v}\\
&\forall 1\leq t\leq d,~i\in I:& & \bx_i\leq \sum_{j=1}^{\ell_t} \by^{t,j}_i \label{eq:fptas_x}\\
&\forall 1\leq t\leq d,~ 1\leq j \leq \ell_t, ~|K^t_j|=1, i\in I\setminus L_{K^t_j,\gamma} : && 
\by^{t,j}_i=0\label{eq:fptas_large}\\ 
& & & \bx\in P(\II) \\
&\forall 1\leq t\leq d,~ 1\leq j \leq \ell_t: &&\by^{t,j}\in P_{K^t_j} 
\end{align}
It can be easily observed that  $D_v\neq \emptyset $ if and only if $v>\OPT$. Furthermore, it holds that if $\left(\bx, \left(\left( \by^{t,j}\right)_{j=1}^{l_t}\right)_{t=1}^{d}\right)\in D_v$ for some $v\geq 0$ then $\bx\in P$. 

Given $v\geq 0$ and $\delta>0$ we can use the ellipsoid method to determine if $D_v=\emptyset$ using a separation oracle. Given $\left(\bx, \left(\left( \by^{t,j}\right)_{j=1}^{l_t}\right)_{t=1}^{d}\right)$, the separation oracle will first verify constraints \eqref{eq:fptas_v}, \eqref{eq:fptas_x} and \eqref{eq:fptas_large} hold. If not, it will return the constraint as a separation hyper-plane. Next, it will use a separation oracle for $\bx\in P(\II)$, we note that such oracles exists for all  constraints $\II$ considered in this paper (see \cite{Sc03} for further details). Finally, for every $1\leq t \leq d$ and $1\leq j \leq \ell_t$ the algorithm of Lemma \ref{lem:sep_primal} will be used with $\by^{t,j}$, the block $K^t_j$ of the MKC $(w_t,B_t,W_t)$ and $\eps$. If it returns a separating vector $\bnu$ it can be used as a separator for the ellipsoid algorithm. If it holds that $(1-\eps)\by^{t,j}\in P_{K^t_j}$ for every $1\leq j \leq \ell_t$ and $1\leq t\leq d$ then the separation oracle aborts the ellipsoid and returns$\left(\bx, \left(\left( \by^{t,j}\right)_{j=1}^{l_t}\right)_{t=1}^{d}\right)$.

As in the proof of Lemma \ref{lem:sep_primal} the ellipsoid algorithm can terminate in one of two ways. Either it determines that $D_v=\emptyset$, or the separation oracle aborts the process and returns $\left(\bx, \left(\left( \by^{t,j}\right)_{j=1}^{l_t}\right)_{t=1}^{d}\right)$. In the latter case it can be easily verified that $(1-\eps)\cdot \left(\bx, \left(\left( \by^{t,j}\right)_{j=1}^{l_t}\right)_{t=1}^{d}\right)\in D_{(1-\eps)v}$.  Furthermore, the running time of the ellipsoid method is polynomial in the 
input size and in $\eps^{-1}$.

Thus, a binary search can be used to find $(1-\eps)v_2<v_1<v_2$ and $\left(\bx, \left(\left( \by^{t,j}\right)_{j=1}^{l_t}\right)_{t=1}^{d}\right)$ such that $D_{v_2}=\emptyset$ and $\left(\bx, \left(\left( \by^{t,j}\right)_{j=1}^{l_t}\right)_{t=1}^{d}\right)\in D_{v_1}$. Hence $\OPT<v_2$, $\bx\in P$ and $$\bx\cdot \bc\geq v_1\geq (1-\eps)v_2 \geq (1-\eps)\OPT.$$
The process can be implemented in time polynomial in $\eps^{-1}$ and the input size , thus the promised FPTAS is obtained.

\end{proof}

\section{Approximation Algorithm: Omitted Proofs}
\label{app:algorithm}
\begin{proof}[Proof of Lemma \ref{lem:instance_soundness}]
	Let $\mathcal{R}$ be the $d$-MKCP instance
	$\left(I,\left(w_t,B_t,W_t\right)_{t=1}^d,\II, f\right)$ and $\left(K^t_j\right)_{j=0}^{\ell_t}$   the $N$-leveled
	partition of $B_t$ for $1\leq t\leq d$. W.l.o.g assume that $A_t$ is disjoint for any $1\leq t\leq d$. That is,  $A_t(b_1)\cap A_t(B_2) =\emptyset$ for any $b_1\neq b_2$, 
	
	For any $1\leq t\leq d$ let $V_t=\{ i\in A(b)~|~0\leq j \leq \ell_t,~b\in K^t_j,~|K^t_j|=1,~w_t(i)>\gamma \cdot W_t(b)\}$ be the set of all item assigned to a block of a single bin in the $t$-th MKC and are $\gamma$-heavy. Since there are at most $N^2$ blocks of a single bin (by the definition of $N$-leveled partition) and a configuration cannot contain $\gamma^{-1}$ or more $\gamma$-heavy items, it follows that $|V_t| \leq \frac{N^2}{\gamma}$. 
	Let $V=\bigcup_{t=1}^{d} V_t$, thus $|V| \leq \frac{N^2\cdot d}{\gamma}$. 
	
	Define $S'=S\setminus V$ and
	denote $f(S)=\OPT$. By the definition of restricted $d$-MKCP it holds that $f(\{i\})-f(\emptyset) \leq \frac{\OPT}{\xi}$ for any $i\in I$. Denote $V=\{i_1,\ldots, i_{\eta}\}$. As $f$ is submodular  we have 
	\begin{equation*}
		\begin{aligned}
	f(S)&=f(S')+\sum_{j=1}^{\eta} f(S'\cup\{i_1, \ldots, i_{j}\})-f(S'\cup \{i_1, \ldots, i_{j-1}\})\\
	&\leq 
	f(S')+\sum_{j=1}^{\eta} f(\{ i_{j}\})-f(\emptyset)\\
	&\leq f(S')+|V|\cdot \frac{\OPT}{\xi}.
	\end{aligned}
	\end{equation*}
Hence, $f(S')\geq \left( 1-\frac{N^2 \cdot d}{\gamma\cdot \xi} \right)f(S)$.

We are left to show that $\one_{S'}\in P$.
For $1\leq t\leq d$ and $0\leq j \leq \ell_t$ define $S'_{t,j}=\bigcup_{b\in K^t_j} A_t(b) \cap S'$, $\by^{t,j}=\one_{S'_{t,j}}$ and 
 $\bz^{t,j}\in [0,1]^{\cC_{K^t_j}}$ (recall $\cC_{K^t_j}$ is the set of  $K^t_j$-configurations) by $\bz^{t,j}_C=1$ if there is $b\in K^t_j$ such that $C=A_t(b)$ and $\bz^{t,j}_C=0$ otherwise.
 It can be easily verified that $(\one_{S'_{t,j}},\bz^{t,j})\in P^e_{K^t_j}$ where $P^e_{K^t_j}$ is the extended block polytope of $K^t_j$. Thus, $\by^{t,j}=\one_{S'_{t,j}}\in P_{K^t_j}$ where $P_{K^t_j}$ is the block polytope of $K^t_j$.
 
 Let $1\leq t\leq d$, $0\leq j \leq \ell_t$ and $i\in I\setminus L_{K^t_j,\gamma}$ with $K^t_j=\{b\}$ for some $b\in B_t$. If $i\in A_t(b)$ it follows that $i\in V_t$, therefore $i\notin S'$ and $i\notin S'_{t.j}$. 
 If $i\notin A_t(b)$ then $i\notin S'_{t,j}$ as well. Hence $\by^{t,j}_i=0$.
  As $\one_{S'} =\sum_{j=0}^{\ell_t} \by^{t,j}$ we can conclude that $(\one_{S'},\by^{t,0},\ldots, \by^{t,\ell_t})\in P^e_t$ for any $1\leq t\leq d$, where $P^e_t$ is the extended $\gamma$-partition polytope of $(w_t, B_t, W_t)$ and the partition $(K^t_j)_{j=0}^{\ell_t}$.
 
 Finally,
  as $S'\subseteq S$ and $S\in \II$ it follows that $S'\in \II$ and 
$\one_{S'}\in P(\II)$. Thus we can conclude $\one_{S'}\in P$.

 \end{proof}

The following definition and Lemma are used in the proof of Lemma~\ref{lem:restricted_approx}. 

We first limit our attention to $d$-MKCP instances solved by the  theorem.
Given a set of items $I$, 
a pair $(\II,f)$  is {\em valid} if $\II\subseteq 2^I$, $f:2^I\rightarrow \mathbb{R}_{\geq 0}$ is submodular and one of the following holds:
\begin{enumerate}
	\item 
	$\II=2^I$ (i.e., $\II$ does not represent a constraint).
	\item The function $f$ is  monotone and  $\II$ defines the independent sets of a matroid.
	\item The function $f$ is modular and $\II$ defines
	the independent sets of a matroid,
	the intersection of the independent sets of two matroids, or a matching.
\end{enumerate}
If follows that Lemma~\ref{lem:restricted_approx} only deals with $d$-MKCP instances in which $(f,\II)$, the function and additional constraint, are valid. 

We use the following lemma to provide a unified probabilistic analysis which does not depend on the sampling method used in Step~\ref{restricted:sampling} of Algorithm~\ref{alg:restricted}.

\begin{lemma}
	\label{lem:sampling}
	Let $(\II ,f)$ be a valid pair,  $\bx \in P(\II)$ (recall $P(\II)$ is the convex hull of $\II$) and $\delta>0$. Also,
	let $R$ be a random set sampled by $\bx$, $\delta$, and $\II$. Then,
	\begin{enumerate}
		\item For any $i\in I$, $\Pr(i\in R) = (1-\delta)^2 \bx_i$.
		\item For any $\ba\in [0,1]^I$,  $\eps>0$  and $\zeta \geq (1-\delta)^2 \ba \cdot \bx$,
		$$\Pr\left(\sum_{i\in R} \ba_i \geq (1+\eps) \zeta \right)\leq \exp\left( -\frac{\zeta\cdot \delta\cdot \eps^2}{20}\right).$$
		\item Let $\eps>0$, $f_{\max}= \max_{i\in I} f(\{i\})-f(\emptyset)$ and $\zeta \leq (1-\delta)^2 F(\bx)$, where $F$ is the multilinear extension of $f$. Also, assume $f$ is monotone.  
		Then,
		$$\Pr\left(f(R) \leq (1-\eps)\zeta \right)\leq \exp\left(- \frac{\zeta\cdot  \delta\cdot \eps^2}{20 \cdot f_{\max}}\right).$$
		\item If $\II=2^I$ then the events $(i\in R)_{i\in I}$ are independent.
		\item The sampling of $R$  can be implemented in polynomial time.
	\end{enumerate}
\end{lemma}
Lemma \ref{lem:sampling} follows directly from \cite{CVZ10}, \cite{CVZ11}, the concentration bound of \cite{Vo10}  and the Chernoff bounds of \cite{RKPS06}. The requirement for $(\II,f)$ to be valid stems from the limitation of the concentration bound in \cite{Vo10}, \cite{CVZ10} and \cite{CVZ11} respectively.

Before we proceed to the proof of Lemma~\ref{lem:restricted_approx} we present its main structure. 
We use Lemma \ref{lem:instance_soundness} to show that there is  $S'\subseteq I$ such that $\one_{S'}\in P$ and $f(S')\geq \left(1- \frac{N^2\cdot d}{\xi \gamma}\right)\OPT \geq (1	-\eps^2)\OPT$, where the last inequality follows from a proper selection of the parameters. Thus, it holds that $F(\bx)\geq c\cdot(1-\eps^2)^2\OPT$. 
We use standard concentration bounds to show that 
with high probability (w.h.p)
$J\cap I^t_j$ satisfies the conditions of Fractional Grouping (Lemma~\ref{lem:grouping}); thus, w.h.p. the packing in Step~\ref{restricted:packing} succeeds.
In the monotone (and modular) case, concentration bounds are used to show that $f(J)$ is sufficiently large w.h.p. In the non-monotone case, the expected value of the returned solution is bounded using the  FKG inequality (as in \cite{CVZ14}).

\begin{proof}[Proof of Lemma \ref{lem:restricted_approx}]
	
Let $d\in \mathbb{N}$, $\eps>0$ and $M>0$. W.l.o.g assume $\eps<0.1$. Define $\delta =\eps^2$ and $\aeps=\frac{\delta}{4}$ (as in Algorithm \ref{alg:restricted}).  Let $\rho>0$ be a constant such that the algorithm of \cite{KK82} returns a bin packing using $\OPT+\rho \log^2\OPT$ bins for every bin packing instance.  By the monotone convergence theorem it holds that 
$$\lim_{N\rightarrow \infty}\sum_{j=1}^{\infty} N^2 \left( 16 \cdot \delta^{-2 } + 2 \right)\cdot\exp  \left(    -\frac{\delta^4}{640}   \cdot N^j \right)=0.$$
Thus there is  $N>M$ such that $ 
-\frac{\aeps}{2}N' + 4 \cdot 4^{\aeps^{-2}}   + \rho \log^2 N' \leq 0 $ for any $N'\geq N$, $N> 32 \cdot \delta^{-2}$  
and 
\begin{equation}
	\label{eq:N} \sum_{j=1}^{\infty} N^2 \left( 16 \cdot \delta^{-2 } + 2 \right)\cdot\exp  \left(    -\frac{\delta^4}{640}   \cdot N^j \right)< \frac{\eps^2}{4\cdot d}.
\end{equation}

Select $0<\gamma<\frac{\delta}{4}$ such that  
\begin{equation}
	\label{eq:gamma_select}
	N^2 \cdot \exp\left(-\frac{\delta^3}{160} \gamma^{-1} \right)<\frac{\eps^2}{4 \cdot d}.
	\end{equation}
Also, select $\xi\in \mathbb{N}$ such that 
$\exp\left(-\frac{0.1\cdot (1-\eps^2)^4 \eps^6 }{20} \cdot \xi \right)\leq \frac{\eps^2}{2}$ and $\frac{N^2\cdot d}{\gamma \cdot \xi}\leq \eps^2$. 

Let $\mathcal{R}$  be an $(N,\xi)$-restricted   $d$-MKCP instance defined by  $\left(I,\left(w_t,B_t,W_t\right)_{t=1}^d,\II, f\right)$, such that  $(f,\II)$ is valid, and $(K^t_j)_{j=0}^{\ell_t}$ the $N$-leveled partition of $B_t$ for $1\leq t\leq d$. Let $\OPT$ be value of the optimal solution for $\mathcal{R}$. By Lemma \ref{lem:instance_soundness} there is $S'\subseteq I$ such that $\one_{S'}\in P$ and $f(S')\geq \left(1- \frac{N^2\cdot d}{\gamma \cdot \xi}\right)\OPT \geq (1-\eps^2) \OPT$, where $P$ is the $\gamma$-instance polytope of $\mathcal{R}$.

To execute Step \ref{restricted:opt} of Algorithm \ref{alg:reduction},
the algorithms of \cite{BF19,CCPV07} can be used to obtain $\bx\in P$ such that $F(\bx)\geq c\cdot (1-\eps^2) f(S') $ ($c$ is defined in the statement of the lemma). The algorithms
require an optimization oracle which given $\bc \in \mathbb{R}^I$, the oracle returns $\bbeta\in P$ such that $\bc \cdot \bbeta$ is maximized.  However, the algorithms can be easily
adapted to use the FPTAS given in Lemma \ref{lem:instance_fptas} instead of the optimization oracle, while preserving the polynomial 
running time and the approximation guarantees. In case $f$ 
is modular  Lemma \ref{lem:instance_fptas} can be used directly
to find $\bx \in P$ s.t $F(\bx)\geq (1-\eps^2) f(S')$ as the multilinear extension $F$ in linear when $f$ is linear  (see Lemma~\ref{lem:multilinear_of_linear}). Thus, we always have that 
\begin{equation}
	\label{eq:Fx_lb}
	F(\bx)\geq c\cdot(1-\eps^2) f(S')\geq c\cdot(1-\eps^2)^2 \OPT
\end{equation}

As in the algorithm, we use $\aeps =\frac{\delta}{4}$. For any  $1\leq t\leq d$ and $0\leq j \leq \ell_t$  let $G^{t,j}_1,\ldots, G^{t,j}_{\tau_{t,j}}$ be the $\aeps$-grouping of $(1-\delta)\by^{t,j}$ with respect to the block $K^t_j$ (Lemma  \ref{lem:grouping}).
Note that $I^t_j$ and $\by^{t,j}$ are defined in Steps \ref{restricted:defs} and \ref{restricted:association} of the algorithm.

We utilize the following definition.
\begin{defn}
For $1\leq t\leq d$ and $0\leq j\leq \ell_t$, a subset $Q\subseteq \supp(\bx)$ is {\em $(t,j)$-compliant} the following holds.
\begin{enumerate}
	\item If $j<N^2$ (i.e, $j$ for which $|K^t_j|=1$) then $w_t(Q\cap I^t_j)\leq W^*_{K^t_j}$. 
	\item If $j\geq N^2$ then 
	\begin{equation*}
	\begin{array}{c}\displaystyle \forall 1\leq k \leq \tau_{t,j}:~~~ |Q\cap G^{t,j}_k \cap I^t_j|\leq \aeps |K^t_j|  \textnormal{ ~~~~and } \\ 
		\\
		\displaystyle w_t(Q\cap I^t_j \cap L_{K^t_j,\aeps}) \leq \sum_{i\in L_{K^t_j, \aeps} } (1-\delta)\by^{t,j}_i \cdot w_t(i) + \frac{\aeps}{4} W^*_{K^t_j}\cdot |K^t_j| 
	\end{array}
	\end{equation*}
\end{enumerate}
Furthermore, we say $Q$ is {\em compliant} if it is $(t,j)$-compliant for every $1\leq t\leq d$ and $0\leq j\leq \ell_t$. 
\end{defn}

We proceed with the following claims.
\begin{claim}
	\label{claim:psi_suff}
	For every $1\leq t\leq d$ and 
	$0\leq j \leq \ell_t$,
	if $R$ is $(t,j)$-compliant then  Algorithm \ref{alg:restricted} succeeds in packing $J\cap I^t_j$ into $K^t_j$ in Step \ref{restricted:packing}.
\end{claim}
\begin{proof}
	If $ j<N^2$ then $|K^t_j| =1$. Since $R$ is  $(t,j)$-compliant  the items of $R\cap I^t_j$ fit into the single bin in $K^t_j$. As $J\subseteq R$ it follows that the items of $J\cap I^t_j$ fit into  the single bin in $K^t_j$ as well. 
	
	If $N^2\leq j $ it holds that $|K^t_j|\geq N $.
	By Lemma \ref{lem:grouping}, since $R$ is $(t,j)$-compliant, it follows that $R\cap I^t_j$ can be packed into at most
	$$ (1-\aeps)|K^t_j| + 4 \cdot 4^{\aeps^{-2}} + 2 \cdot \frac{\aeps}{4} |K^t_j|=\left(1-\frac{\aeps}{2}\right)|K^t_j| + 4 \cdot 4^{\aeps^{-2}} \leq |K^t_j|$$ 
	bins of capacity $W^*_{K^t_j}$. Since $J\cap I^t_j\subseteq R\cap I^t_j$ it holds that the items of $J\cap I^t_j$ can also be packed into the above number of bins.  Therefore, the bin packing algorithm of \cite{KK82} packs the items into  at most
	$$ \left(1-\frac{\aeps}{2}\right)|K^t_j| + 4 \cdot 4^{\aeps^{-2}}  + \rho \log^2 |K^t_j| \leq |K^t_j|$$
bins of capacity $W^*_{K^t_j}$.\
\end{proof}

For every $1\leq t\leq d$ and $0\leq j\leq \ell_t$ we define an event $\Psi_{t,j} = \{\textnormal{$R$ is $(t,j)$-compliant}\}$. 
\begin{claim}
	\label{claim:psi_prob}
	
	$$
	\Pr\left( \exists 1\leq t\leq d \textnormal{ and } 1\leq j \leq \ell_t :~\neg \Psi_{t,j}\right)<\frac{\eps^2}{2}$$
\end{claim}
As the proof of the lemma is technical, we first sketch an overview of the proof. 
It holds that each of the events $\Psi_{t,j}$ consists of a few linear constraints over $R$, of the form $\sum_{i\in R} \ba_i < L$, for some $L>0$ and $\ba\in \mathbb{R}_{\geq 0}^I$. For each of these constraints, we can show that $\E[\sum_{i\in R}\ba_i]$ is strictly smaller than $L$. This property follows from the Block Association  and from the distribution by which items are sampled (i.e., $\Pr(i\in R) = (1-\delta)^2 \bx$).  As the instance is restricted and $\bx\in P$, it also follows that the ratio $\frac{\ba_i}{L}$ is always bounded, and thus the concentration bounds of Lemma~\ref{lem:sampling} can be used to show a constraint is violated only with an exponentially small probability. The union bound is then used to show the probability that one of these constraints is violated is bounded by $\frac{\eps^2}{2}$. 
\begin{proof}[Proof of Lemma~\ref{claim:psi_prob}]
	For $1\leq t\leq d $ and 
	 $0\leq j\leq \min\{N^2-1, \ell_t\}$, 
	by Lemma \ref{lem:association} it holds that there is $i^*\in I$ such that 
	$$\sum_{i\in I^t_j\setminus \{i^*\}} (1-\delta) \bx_i \cdot w_t(i) \leq \sum_{i\in I} (1-\delta)\by^{t,j}_i \cdot w_t(i)\leq (1-\delta) W^*_{K^t_j}.$$ Since $\Pr(i\in R )=(1-\delta)^2 \bx_i$ for any $i\in I$, it follows that
	$$\E\left[ w_t(R\cap I^t_j\setminus \{i^*\}) \right] 
	= \sum_{i\in I^t_j\setminus \{i^*\}} (1-\delta)^2 \bx \cdot w_t(i )\leq (1-\delta)^2 W^*_{K^t_j}.$$
	As $I^t_j \subseteq \supp\{\by^{t,j}\}\subseteq \{i\in I ~|~w_t(i)\leq \gamma \cdot W^*_{K^t_j}\}$ (due to Lemma \ref{lem:association} and Definition  \ref{def:extended_gamma_partition}), and $\gamma<\frac{\delta}{2}$ it follows that
	$\E\left[ w_t(R\cap I^t_j) \right]  \leq \left(1-\frac{\delta}{2}\right)\cdot W^*_{K_j}$. Hence,
	$\E\left[ \sum_{i\in R} {\one_{i\in I^t_j} \cdot w_t(i)}\frac{1}{\gamma\cdot W^*_{K^t_j}} \right]  \leq \left(1-\frac{\delta}{2}\right)\cdot \frac{1}{\gamma}$.
	\footnote{For $i\in I$ and $S\subseteq I$ we use $\one_{i\in S}=1$ if $i\in s$ and $\one_{i\in S}=0$ otherwise.} Following the same argument and by Lemma \ref{lem:sampling} it holds that
	\begin{equation}
		\label{eq:psi_singlton}
		\begin{aligned}
			\Pr&(\neg \Psi_{t,j}) = \Pr( w_t(R\cap I^t_j) >W^*_{K^t_j}) =
			\Pr\left(\sum_{i\in R} \frac{\one_{i\in I^t_j} \cdot w_t(i)}{\gamma\cdot W^*_{K^t_j}} >\gamma^{-1} \right)
			\\&
			\leq \Pr\left(\sum_{i\in R} \frac{\one_{i\in I^t_j} \cdot w_t(i)}{\gamma\cdot W^*_{K^t_j}} \geq \left( 1+\frac{\delta}{2}\right)\left(1-\frac{\delta}{2}\right)\cdot \frac{1}{\gamma} \right)
			\leq 
			\exp\left( -\frac{\delta^3}{80}
			\left(1-\frac{\delta}{2}\right)\cdot \gamma^{-1}
			\right)
			\\&
			\leq \exp\left( -\frac{\delta^3}{160}
			\cdot \gamma^{-1}
			\right).
		\end{aligned}
	\end{equation}
	
	For $1\leq t\leq d$ and $N^2\leq  j \leq \ell_t$, by Lemmas \ref{lem:association}  and \ref{lem:sampling} for any $1\leq k \leq \tau_{t,j}$ it holds that
	$$\E\left[ |R\cap I^t_j \cap G^{t,j}_k|  \right] \leq (1-\delta )\left(\aeps |K^t_j| +2 \right)\leq \left( 1-\frac{\delta}{2}\right)\aeps |K^t_j|,$$
	where the last transition holds since $|K^t_j|\geq N$ and thus $|K^t_j| >32\cdot\delta^{-2} =8\cdot \delta^{-1}  \cdot \aeps^{-1}$. 
	Hence, following Lemma \ref{lem:sampling} it holds that
	\begin{equation}
		\label{eq:partition_prob}
		\begin{aligned}
			\Pr&\left( |R\cap I^t_j \cap G^{t.j}_k|>\aeps |K| \right)
			\leq
			\Pr\left( \sum_{i\in R} \one_{i\in I^t_j\cap G^{t,j}_k} \geq \left(1+\frac{\delta}{2}\right) \left( 1-\frac{\delta}{2}\right)\aeps |K^t_j|\right)\\
			&\leq  \exp\left(-\frac{\delta^3}{80} \cdot\left(1- \frac{\delta}{2}\right) \aeps |K^t_j|\right) = \exp\left(-\frac{\delta^4}{640}\cdot   |K^t_j|\right).
		\end{aligned}
	\end{equation}
	Similarly, by Lemma \ref{lem:association} there is $i^*\in I$ such that $\sum_{i\in I^t_j \cap L_{K^t_j, \aeps} \setminus \{i^*\}} (1-\delta) \bx_i \cdot w_t(i) \leq \sum_{i\in L_{K^t_j, \aeps}} (1-\delta) \by^{t,j}_i \cdot w_t(i)$. Thus, $$\sum_{i\in I^t_j \cap L_{K^t_j, \aeps} } (1-\delta) \bx_i \cdot w_t(i) \leq \sum_{i\in L_{K^t_j, \aeps}} (1-\delta) \by^{t,j}_i \cdot w_t(i) +\aeps W^*_{K^t_j} \leq 
	\sum_{i\in L_{K^t_j, \aeps}} (1-\delta) \by^{t,j}_i \cdot w_t(i) +\frac{\aeps}{4} |K^t_j| W^*_{K^t_j},$$
	where the last transition requires $|K^t_j|\geq 4$ which  holds due to the parameter selection. Let $\zeta  = \sum_{i\in L_{K^t_j, \aeps}} (1-\delta) \by^{t,j}_i \cdot w_t(i) +\frac{\aeps}{4}\cdot |K^t_j| \cdot W^*_{K^t_j}$. Subsequently, by Lemma \ref{lem:sampling} we have $\E\left[ w_t(I_j\cap L_{K_j,\aeps} \cap R)\right] \leq (1-\delta)\zeta$ and 
	\begin{equation}
		\label{eq:slack_prob}
		\begin{aligned}
			\Pr&\left(  w_t(I^t_j \cap L_{K^t_j, \aeps} \cap R) > \zeta\right) \leq  \Pr \left( \sum_{i\in R} \one_{i\in I^t_j \cap L_{K^t_j, \aeps }} \cdot \frac{w_t(i)}{W^*_{K^t_j}} \geq (1+\delta)(1-\delta )\frac{\zeta}{W^*_{K^t_j}} \right)\\
			&\leq \exp  \left(    -\frac{\delta^3}{20} \cdot (1-\delta) \cdot  \frac{\zeta}{W^*_{K^t_j}} \right)\leq 
			\exp  \left(    -\frac{\delta^4}{640}   \cdot  |K^t_j| \right).
		\end{aligned}
	\end{equation}
	The last inequality used $\zeta \geq \frac{\aeps}{4} \cdot |K^t_j| \cdot W^*_{K^t_j}= \frac{\delta}{16} \cdot |K^t_j| \cdot W^*_{K^t_j}$.
	By \eqref{eq:partition_prob} and \eqref{eq:slack_prob}, it follows that 
	\begin{equation}
		\label{eq:psi_multi}
		\begin{aligned}
			\Pr(\neg \Psi_{t,j} ) &\leq (1+\tau_{t,j}) \exp  \left(    -\frac{\delta^4}{640}   \cdot  |K^t_j| \right) \leq 
			\left(16\cdot \delta^{-2}+2\right) \exp  \left(    -\frac{\delta^4}{640}   \cdot  |K^t_j| \right),
		\end{aligned}
	\end{equation}
	where the second inequality follows from $\tau_{t,j} \leq \aeps^{-2}+1 = 16\cdot \delta^{-2}+1$ (Lemma \ref{lem:grouping}). 
	
	By \eqref{eq:psi_singlton} and \eqref{eq:psi_multi}, it holds that 
	\begin{equation*}
		\begin{aligned}
		\Pr&(\exists 1\leq t\leq d \textnormal{ and }0\leq j\leq \ell_t:~\neg \Psi_{t,j}) \\ 
		&\leq 
		d\cdot N^2 \exp\left(-\frac{\delta^3}{160} \gamma^{-1} \right)
		+d\cdot \sum_{j=1}^{\infty} N^2 \left( 16 \cdot \delta^{-2 } + 2 \right)\cdot\exp  \left(    -\frac{\delta^4}{640}   \cdot N^j \right) <\frac{\eps^2}{2},
	\end{aligned} \end{equation*}
	where the last inequality is by \eqref{eq:gamma_select} and \eqref{eq:N}. 
\end{proof}

\begin{claim}
	\label{claim:monotone}
If $f$ is monotone (or modular) then the expected value of the returned solution is at least $c\cdot (1-\eps)\OPT$.
\end{claim}
\begin{proof}
Since $f$ is monotone we have $J=\eta_f(R)=R$. 
	
Denote $f_{\max}=\max_{i\in I} f(\{i\})-f(\emptyset)$. Since $\mathcal{R}$ is $(N,\xi)$-restricted  $d$-MCKP instance it holds that $f_{\max}\leq \frac{\OPT}{\xi}$. Using \eqref{eq:Fx_lb} and  Lemma \ref{lem:sampling}, it holds that 
	\begin{equation}
		\begin{aligned}
			\label{eq:fR_concentration}
			\Pr&\left(f(R)< c\cdot (1-\eps^2)^5 \cdot \OPT \right)\leq \exp\left( -\frac{ (1-\eps^2)^4\cdot \eps^2 \cdot \eps^4\cdot c \cdot \OPT }{20 \cdot f_{\max}}\right)\\
			&\leq 
			\exp\left( -\frac{ (1-\eps^2)^4 \cdot \eps^6\cdot c }{20}\xi\right)\leq \frac{\eps^2}{2}.
		\end{aligned}
	\end{equation}
	
Let $\Phi$ be the event in which $\Psi_{t,j}$ occurs for every $1\leq t \leq d$ and $0\leq j \leq \ell_t$ , and $f(R)\geq c\cdot (1-\eps^2)^5 \cdot \OPT$.  By Claim \ref{claim:psi_prob}  and \eqref{eq:fR_concentration} it holds that $\Pr(\Phi)\geq (1-\eps^2)$. By Claim \ref{claim:psi_suff} if $\Phi$ occurs then the algorithm returns $J=R$ as the solution along with corresponding assignments to the MKCs (Step \ref{restricted:return} of the algorithm). Therefore, the expected value of the solution returned by the algorithm at least
\begin{equation*}
\Pr(\Phi)\cdot \E[f(R)~|~\Phi]\geq (1-\eps^2)\cdot c \cdot (1-\eps^2)^5 \OPT \geq c\cdot(1-\eps)\OPT
\end{equation*}
\end{proof}

Claim \ref{claim:monotone} completes the proof for the case that $f$ is monotone. We are left to handle the case in which $f$ is non-monotone. As we assume $(\II,f)$ is a valid pair, it follows that $\II=2^I$. 

In the following we show that our algorithm behaves similarly to a contention resolution scheme, and prove the approximation guarantee similarly to \cite{CVZ14}. It seems possible to use the theorems of \cite{CVZ14} directly to show the correctness, instead of adjusting their proofs to our setting. However, the technical overhead required to present our procedure in a way which  matches the definitions in \cite{CVZ14} appears to surpass the benefits of such an approach. The technical overhead stems from the fact that our rounding depends on the vectors $\by^{t,j}$ and not just on $\bx$, the rounding  does not strictly adhere to the definition of monotonicity in \cite{CVZ14}, and since the polytope used is only an approximate representation of the problem.

The following is a variant of Claim \ref{claim:psi_prob}, and is proved similarly.

\begin{claim}
	\label{claim:psi_prob_with_Q}
	Assume $\II=2^I$. Then 
	for any $Q\subseteq \supp(\bx)$, $|Q|\leq 1$ it holds that 
	$$
	\Pr\left( \exists 1\leq t\leq d \textnormal{ and } 1\leq j \leq \ell_t :~\neg \Psi_j~|~Q\subseteq R\right)<\frac{\eps^2}{2}$$
\end{claim}
\begin{proof}
	As $\II=2^I$, by Lemma \ref{lem:sampling}, it follows that the events $(i\in R)_{i\in I}$ are independent. 
	
	For $1\leq t\leq d $ and 
	$0\leq j\leq \min\{N^2-1, \ell_t\}$, 
	by Lemma \ref{lem:association} it holds that there is $i^*\in I$ such that 
	$$\sum_{i\in I^t_j\setminus \{i^*\}} (1-\delta) \bx_i \cdot w_t(i) \leq \sum_{i\in I} (1-\delta)\by^{t,j}_i \cdot w_t(i)\leq (1-\delta) W^*_{K^t_j}.$$ Since $\Pr(i\in R )=(1-\delta)^2 \bx_i$ for any $i\in I$, it follows that
	$$\E\left[ w_t(R\cap I^t_j\setminus \{i^*\}\setminus Q)\right] 
	= \sum_{i\in I^t_j\setminus \{i^*\}\setminus Q} (1-\delta)^2 \bx \cdot w_t(i )\leq
 (1-\delta)^2 W^*_{K^t_j}.$$
	As $I^t_j \subseteq \supp\{\by^{t,j}\}\subseteq \{i\in I ~|~w_t(i)\leq \gamma \cdot W^*_{K^t_j}\}$ (due to Lemma \ref{lem:association} and Definition  \ref{def:extended_gamma_partition}),  and $\gamma<\frac{\delta}{4}$ it follows that,
	$$\E\left[ w_t(R\cap I^t_j\setminus Q) \right]  \leq  \gamma \cdot  W^*_{K^t_j} + (1-\delta)^2 W^*_{K^t_j}\leq \left(1-\frac{3\cdot\delta}{4}\right)\cdot W^*_{K^t_j}.$$
 Hence,
	$\E\left[ \sum_{i\in R\setminus Q} {\one_{i\in I^t_j} \cdot w_t(i)}\frac{1}{\gamma\cdot W^*_{K^t_j}} \right]  \leq \left(1-\frac{3\cdot\delta}{4}\right)\cdot \frac{1}{\gamma}$.
Recall that as $\II=2^I$,  by Lemma~\ref{lem:sampling},  it holds that the sum $\sum_{i\in R\setminus Q} {\one_{i\in I^t_j} \cdot w_t(i)}\frac{1}{\gamma\cdot W^*_{K^t_j}}$  is independent of the event  $Q\subseteq R$. Thus,
	\begin{equation}
		\label{eq:psi_singleton_with_Q}
		\begin{aligned}
			\Pr&(\neg \Psi_{t,j}~|~Q\subseteq R) = \Pr( w_t(R\cap I^t_j\setminus Q) >W^*_{K^t_j}-w(I^t_j\cap Q) ~|~Q\subseteq R)\\&
			 \leq
			\Pr\left(\sum_{i\in R\setminus Q} \frac{\one_{i\in I^t_j} \cdot w_t(i)}{\gamma\cdot W^*_{K^t_j}} >\gamma^{-1}\left(1-\frac{ \delta}{4} \right)\right)
			\\&
			\leq \Pr\left(\sum_{i\in R\setminus Q} \frac{\one_{i\in I^t_j} \cdot w_t(i)}{\gamma\cdot W^*_{K^t_j}} \geq \left( 1+\frac{\delta}{2}\right)\left(1-\frac{3\cdot \delta}{4}\right)\cdot \frac{1}{\gamma}  \right)
			\\&
			\leq 
			\exp\left( -\frac{\delta^3}{80}
			\left(1-\frac{3\cdot \delta}{4}\right)\cdot \gamma^{-1}
			\right)
			\\&
			\leq \exp\left( -\frac{\delta^3}{160}
			\cdot \gamma^{-1}
			\right),
		\end{aligned}
	\end{equation}
where first equality follows from the definition of $\Psi_{t,j}$, the first inequality uses $|Q|\leq 1$ and $\gamma<\frac{\delta}{4}$ as well as the independence of the sum from $Q\subseteq R$, the third inequality is by Lemma~\ref{lem:sampling}.
	
	For $1\leq t\leq d$ and $N^2\leq  j \leq \ell_t$, by Lemmas \ref{lem:association}  and \ref{lem:sampling} for any $1\leq k \leq \tau_{t,j}$ it holds that
	$$\E\left[ |R\cap I^t_j \cap G^{t,j}_k\setminus Q|  \right] \leq (1-\delta )\left(\aeps |K^t_j| +2 \right)\leq \left( 1-\frac{3\cdot \delta}{4}\right)\aeps |K^t_j|,$$
	where the last transition holds since $|K^t_j|\geq N$ and thus $|K^t_j| >32\cdot\delta^{-2} =8\cdot \delta^{-1}  \cdot \aeps^{-1}$. 
	Hence, following an argument similar to that of \eqref{eq:psi_singleton_with_Q}, we have, 
	\begin{equation}
		\label{eq:partition_prob_with_Q}
		\begin{aligned}
			\Pr&\left( |R\cap I^t_j \cap G^{t,j}_k|>\aeps |K|~\middle|~Q\subseteq R \right)\\
			&= \Pr\left( |R\cap I^t_j \cap G^{t,j}_k\setminus Q|>\aeps |K| -|Q\cap I^t_j \cap G^{t,j}_k|~\middle|~Q\subseteq R \right)\\
			&\leq 
			\Pr\left( |R\cap I^t_j \cap G^{t,j}_k\setminus Q|>\aeps |K| \left(1-\frac{\delta}{4} \right)\right)
			\\
			&\leq
			\Pr\left( \sum_{i\in R\setminus Q} \one_{i\in I^t_j\cap G^{t,j}_k} \geq \left(1+\frac{\delta}{2}\right) \left( 1-\frac{3\cdot \delta}{4}\right)\aeps |K^t_j|\right)\\
			&\leq  \exp\left(-\frac{\delta^3}{80} \cdot\left(1- \frac{\delta}{2}\right) \aeps |K^t_j|\right) = \exp\left(-\frac{\delta^4}{640}\cdot   |K^t_j|\right).
		\end{aligned}
	\end{equation}
	Similarly, by Lemma \ref{lem:association} there is $i^*\in I$ such that $$\sum_{i\in I^t_j \cap L_{K^t_j, \aeps} \setminus \{i^*\}} (1-\delta) \bx_i \cdot w_t(i) \leq \sum_{i\in L_{K^t_j, \aeps}} (1-\delta) \by^{t,j}_i \cdot w_t(i).$$
Define  $\zeta  = \sum_{i\in L_{K^t_j, \aeps}} (1-\delta) \by^{t,j}_i \cdot w_t(i) +\frac{\aeps}{4}\cdot |K^t_j| \cdot W^*_{K^t_j}$. Thus, 
\begin{equation*}
	\begin{aligned}
	\sum_{i\in I^t_j \cap L_{K^t_j, \aeps} } &(1-\delta) \bx_i \cdot w_t(i) \leq \sum_{i\in L_{K^t_j, \aeps}} (1-\delta) \by^{t,j}_i \cdot w_t(i) +\aeps W^*_{K^t_j} 
	\\ 
	&\leq 
	\sum_{i\in L_{K^t_j, \aeps}} (1-\delta) \by^{t,j}_i \cdot w_t(i) +\frac{\aeps}{8} |K^t_j| W^*_{K^t_j}=\zeta - \frac{\aeps}{8} |K^t_j| W^*_{K^t_j},
	\end{aligned}
	\end{equation*}
	where the last inequality requires $|K^t_j|\geq 8$ which  holds due to the parameter selection. Subsequently, by Lemma \ref{lem:sampling} we have $\E\left[ w_t(I_j\cap L_{K_j,\aeps} \cap R\setminus Q)\right] \leq (1-\delta)\left(\zeta-\frac{\aeps}{8} |K^t_j| W^*_{K^t_j} \right)$ and 
	\begin{equation}
		\label{eq:slack_prob_with_Q}
		\begin{aligned}
			\Pr&\left(  w_t(I^t_j \cap L_{K^t_j, \aeps} \cap R) > \zeta~\middle|~ Q\subseteq R\right)\\
			&= 
			\Pr\left(  w_t(I^t_j \cap L_{K^t_j, \aeps} \cap R\setminus Q) > \zeta-w_t(I^t_j \cap L_{K^t_j, \aeps} \cap Q )~\middle|~ Q\subseteq R\right)\\
			&\leq 
			\Pr\left(  w_t(I^t_j \cap L_{K^t_j, \aeps} \cap R\setminus Q) > \zeta-\aeps W^*_{K^t_j}\right)
			 \\
			 &\leq 
			 \Pr\left(  w_t(I^t_j \cap L_{K^t_j, \aeps} \cap R\setminus Q) > \zeta-\frac{\aeps}{8}|K^t_j|\cdot W^*_{K^t_j}\right)
			 \\
			& \leq  \Pr \left( \sum_{i\in R\setminus Q} \one_{i\in I^t_j \cap L_{K^t_j, \aeps }} \cdot \frac{w_t(i)}{\aeps W^*_{K^t_j}} \geq (1+\delta)(1-\delta )\frac{\zeta-\frac{\aeps}{8} |K^t_j| \cdot W^*_{K^t_j}}{\aeps W^*_{K^t_j}} \right)\\
			&\leq \exp  \left(    -\frac{\delta^3}{20} \cdot (1-\delta) \cdot  \frac{\zeta-\frac{\aeps}{8} \cdot |K^t_j| \cdot W^*_{K^t_j} }{\aeps  W^*_{K^t_j}} \right)\leq 
			\exp  \left(    -\frac{\delta^4}{640}   \cdot  |K^t_j| \right).
		\end{aligned}
	\end{equation}
	The last inequality used $\zeta \geq \frac{\aeps}{4} \cdot |K^t_j| \cdot W^*_{K^t_j}$.
	By \eqref{eq:partition_prob_with_Q} and \eqref{eq:slack_prob_with_Q}, it follows that 
	\begin{equation}
		\label{eq:psi_multi_with_Q}
		\begin{aligned}
		\Pr(\neg \Psi_{t,j} ~|~Q\subseteq R) &\leq (1+\tau_{t,j}) \exp  \left(    -\frac{\delta^4}{640}   \cdot  |K^t_j| \right) \leq 
			\left(16\cdot \delta^{-2}+2\right) \exp  \left(    -\frac{\delta^4}{640}   \cdot  |K^t_j| \right),
		\end{aligned}
	\end{equation}
	where the second inequality follows from $\tau_{t,j} \leq \aeps^{-2}+1 = 16\cdot \delta^{-2}+1$ (Lemma \ref{lem:grouping}). 
	
	By \eqref{eq:psi_singleton_with_Q} and \eqref{eq:psi_multi_with_Q}, it holds that 
	\begin{equation*}
		\begin{aligned}
			\Pr&(\exists 1\leq t\leq d \textnormal{ and }0\leq j\leq \ell_t:~\neg \Psi_{t,j}~|~Q\subseteq R) \\ 
			&\leq 
			d\cdot N^2 \exp\left(-\frac{\delta^3}{160} \gamma^{-1} \right)
			+d\cdot \sum_{j=1}^{\infty} N^2 \left( 16 \cdot \delta^{-2 } + 2 \right)\cdot\exp  \left(    -\frac{\delta^4}{640}   \cdot N^j \right) <\frac{\eps^2}{2},
	\end{aligned} \end{equation*}
	where the last inequality is by \eqref{eq:gamma_select} and \eqref{eq:N}. 

\end{proof}

W.l.o.g we assume $I=\{1,2,\ldots, n\}$ and that the order in which $\eta_f$ iterates over the items is $1,2,\ldots, n$. Denote $[i]=\{1,2,\ldots, i\}$ for $i\in I$, $[0]=\emptyset$, and $$\forall Q\subseteq I:~~~~\kappa(Q)=\begin{cases} 1 & \textnormal{$Q$ is compliant} \\0 & \textnormal{otherwise}
\end{cases}.$$

\begin{claim}
	\label{claim:FKG}
	Assume $\II=2^I$. 
For any  $i\in I$ it holds that 
$$\E\left[ \kappa(R) \cdot \left( f(J\cap[i])-f(J\cap [i-1])\right)\right]\geq (1-\eps^2) \E\left[ f(R\cap[i])-f(R\cap [i-1])\right].$$
\end{claim}
\begin{proof}
Note that if $i\not\in \supp(\bx)$ the claim trivially holds, thus we can assume $i\in \supp(\bx)$ and $\Pr(i\in R)>0$. 
Define $Q=\{i\}$ and $\pi(X)=\Pr\left( R=X\cup\{i\}~|~Q\subseteq R\right)$ for any $X\subseteq I\setminus \{i\}$. It holds that, 
\begin{equation}
	\label{eq:towards_FKG}
	\begin{aligned}
	\E&\left[ \kappa(R) \cdot \left( f(J\cap[i])-f(J\cap [i-1])\right)\right]=
		\E\left[ \kappa(R) \cdot \one_{i\in J} \cdot f_{J\cap [i-1]}(\{i\})\right]\\
		&\geq\Pr\left(Q\subseteq R \right) \cdot\E\left[ \kappa(R) \cdot \one_{i\in J} \cdot f_{J\cap [i-1]}(\{i\})~\middle|~Q\subseteq R\right]\\
			&\geq\Pr\left(Q\subseteq R \right) \cdot\E\left[ \kappa(R) \cdot \one_{i\in J} \cdot \max\left\{ 0,~f_{J\cap [i-1]}(\{i\})\right\}~\middle|~Q\subseteq R\right]\\
			&=\Pr\left(Q\subseteq R \right)\cdot \E\left[ \kappa(R)  \cdot \max\left\{ 0,~f_{J\cap [i-1]}(\{i\})\right\}~\middle|~Q\subseteq R\right]\\
			&\geq \Pr\left(Q\subseteq R \right) \cdot\E\left[ \kappa(R)  \cdot \max\left\{ 0,~f_{R\cap [i-1]}(\{i\})\right\}~\middle|~Q\subseteq R\right]\\
			&=  \Pr\left(Q\subseteq R \right)  \sum_{X\subseteq I \setminus \{i\}} \Pr\left( R=X\cup\{i\}~|~Q\subseteq R\right)\cdot \kappa(X\cup\{i\}) \cdot 
			\max\left\{ 0,~f_{X\cap [i-1]}(\{i\})\right\} \\
			&=  \Pr\left(Q\subseteq R \right)  \sum_{X\subseteq I \setminus \{i\}} \pi(X)\cdot \kappa(X\cup\{i\}) \cdot 
			\max\left\{ 0,~f_{X\cap [i-1]}(\{i\})\right\},
	\end{aligned}
\end{equation}
where the first inequality holds since $i\in J\subset R$ implies $Q\subseteq R$, the second inequality follows from the definition of $\eta_f$, the second equality holds since if $i\not\in J$ and $Q\subseteq R$  then $\max\{ 0,f_{J\cap[i-1] }(\{i\})\} =0$, the third inequality follows from submodularity and $J\subseteq R$. 

In the following we use the FKG inequality. We refer the reader to the relevant chapter in \cite{AS04b} for the definition of lattice and log-supermodular functions.
Since the events $(i\in R)_{i\in I}$ are independent, it follows that for any $X_1, X_2 \subseteq I\setminus \{i\}$ we have,
$$\pi(X_1) \cdot \pi(X_2) = \pi(X_1\cup X_2)\cdot \pi(X_1\cap X_2).$$
Hence $\pi$ is {\em log supermodular} over the lattice of subsets of $I\setminus \{i\}$. Additionally, the functions  $\kappa(X\cup \{i\})$ and $	\max\left\{ 0,~f_{X\cap [i-1]}(\{i\})\right\} $ is decreasing as functions of $X$. Finally, it holds that $\sum_{X\subseteq I\setminus \{i\}} \pi(X)=1$. Thus, by the FKG inequality (Theorem 2.1 in \cite{AS04b}) it holds that,
\begin{equation}
	\label{eq:FKG_app}
	\begin{aligned}
  \Pr&\left(Q\subseteq R \right)  \sum_{X\subseteq I \setminus \{i\}} \pi(X)\cdot \kappa(X\cup\{i\}) \cdot 
\max\left\{ 0,~f_{X\cap [i-1]}(\{i\})\right\}  \\
&\geq 
\Pr\left(Q\subseteq R \right)  \cdot \left(
 \sum_{X\subseteq I\setminus\{i\} }\pi(X) \cdot \kappa(X\cup\{i\}) 
	\right) 
	\cdot \left(  
	\sum_{X\subseteq I\setminus\{i\} }\pi(X) \cdot 	\max\left\{ 0,~f_{X\cap [i-1]}(\{i\})\right\}
	\right) \\
&= 
\Pr\left(Q\subseteq R \right)\cdot  \E\left[\kappa(R)~\middle|~Q\subseteq R\right]
\cdot \E\left[	\max\left\{ 0,~f_{R\cap [i-1]}(\{i\})\right\}
~\middle|~Q\subseteq R\right]\\
&\geq \left(1-\eps^2\right) \cdot \E\left[\max\left\{ 0,~f_{R\cap [i-1]}(\{i\})\right\} \right] \\
&\geq  \left(1-\eps^2\right) \cdot\E\left[ f(R\cap [i])-f(R\cap [i-1]) \right],
\end{aligned}
\end{equation}
where the second inequality is due to Claim~\ref{claim:psi_prob_with_Q}. The statement of the 
claim is obtained by combining \eqref{eq:towards_FKG} and \eqref{eq:FKG_app}. 
\end{proof}

\begin{claim}
	\label{claim:non_mon}
	If $f$ is non-monotone then the expected value of the solution returned by the algorithm  is at least $c\cdot (1-\eps) \OPT$. 
\end{claim}
\begin{proof}
As $f$ is non-monotone we have $\II=2^I$. 
By Claim \ref{claim:psi_suff} the expected value of the solution returned by the algorithm is at least $\E\left[ \kappa(R) \cdot f(J)\right]$. By Claim~\ref{claim:FKG} we have,
\begin{equation*}
	\begin{aligned}
	\E&\left[ \kappa(R) \cdot f(J)\right]=
		\E\left[ \kappa(R) \left(f(\emptyset  )+\sum_{i=1}^{n}\left( f\left(J\cap[i]\right) - f\left(J\cap[i-1]\right)\right) \right)\right]\\
		&= \E\left[ \kappa(R) \cdot f(\emptyset)  \right] + \sum_{i=1}^{n} \E\left[ \kappa(R) \left( f\left(J\cap[i]\right) - f\left(J\cap[i-1]\right) \right) \right]\\
		&\geq (1-\eps^2)\cdot f(\emptyset)     
		+ \left(1-\eps^2\right) \cdot \sum_{i=1}^{n} \E\left[  f\left(R\cap[i]\right) - f\left(R\cap[i-1]\right)  \right]\\
		&\geq 
		(1-\eps^2)\cdot\E \left[f(R)   \right] \geq (1-\eps) \cdot c \cdot \OPT,
	\end{aligned}
\end{equation*}
where the first inequality is by Claims \ref{claim:psi_prob} and \ref{claim:FKG}, and the last inequality is due to \eqref{eq:Fx_lb}.
\end{proof}
The statement of the lemma follows from Claims \ref{claim:monotone} and \ref{claim:non_mon}.
\end{proof}

\section{Properties of Submodular Functions}
\label{app:sub_prop}

\begin{lemma}
	\label{lem:marginal}
	Let $f:2^I \rightarrow \mathbb{R}_{\geq 0}$ be a submodular function and let $S=\{s_1,\ldots, s_{\ell}\}\subseteq I$, $|S|=\ell$, such that 
	$f(\{s_1,\ldots , s_r\})=\max_{r-1<k\leq \ell} f(\{s_1,\ldots ,s_{r-1}\}\cup \{s_k\})$ for every $1\leq r \leq \ell$.  Also, let $E\in \mathbb{N}$ and $S_E= \{s_r ~|~1\leq r\leq \min\{E, \ell\} \}$. Then for every $i \in S\setminus S_E$ it holds that $f(S_E \cup \{i\}) -f(S_E)\leq \frac{f(S_E)}{E}$. 
	\end{lemma}
\begin{proof}	
	If $E\geq \ell$ then $S\setminus S_E =\emptyset$ and the statement trivially holds.
	Otherwise,  for every $i\in S \setminus S_E$ it holds that
	\begin{equation}
	\begin{aligned}
	f(S_E &\cup\{i\}) -f(S_E) = \frac{1}{E} \sum_{r=1}^{E}\left(
	f(S_E\cup\{i\}) -f(S_E)\right)\\
	&\leq \frac{1}{E} \sum_{r=1}^{E}\left(
	f(\{s_1,\ldots ,s_{r-1}\cup\{i\}) -f(\{s_1,\ldots ,s_{r-1}\})\right) 
	\\
		&\leq \frac{1}{E} \sum_{r=1}^{E}\left(
	f(\{s_1,\ldots ,s_{r-1}\cup\{s_r\}) -f(\{s_1,\ldots ,s_{r-1}\})\right) \\
	&\leq \frac{1}{E} \left(f(S_E)-f(\emptyset) \right) \\
			&\leq \frac{f(S_E)}{E}. 
	\end{aligned}
	\end{equation} 
	The first inequality follows from the submodularity of $f$. 
	The  second inequality follows from $f(\{s_1,\ldots , s_r\})=\max_{r-1<k\leq \ell} f(\{s_1,\ldots ,s_{r-1}\}\cup \{s_k\})$ for every $1\leq r \leq \ell$.
\end{proof}

\begin{lemma}
	\label{lem:submodular_cup}
Let $f:2^I \rightarrow \mathbb{R}$ be a set function and $R\subseteq I$. Define $g:2^I\rightarrow \mathbb{R}$ by $g(S)=f(S\cup R)$ for any $S\subseteq I$. Then,
\begin{enumerate}
	\item If $f$ is submodular then $g$ is submodular.
	\item If $f$ is monotone then $g$ is monotone.
	\item If $f$ is modular then $g$ is modular.
\end{enumerate}
\end{lemma}
\begin{proof}~
	
\begin{enumerate}
	\item Assume $f$ is submodular. Let $S, T \subseteq I$. Then,
	\begin{align*}g(S) +g(T) &= f(R\cup S) +f(R\cup T) \\
	&\geq f \left( (R\cup S )\cup (T \cup R)\right) + f \left( (R\cup S )\cap (T \cup R)\right) \\
	&= f \left( R\cup (S \cup T)\right) + f \left( R\cup (S \cap T)\right)\\
	&=g(S\cup T)+ g(S\cap T )
	\end{align*}
	Thus $g$ is submodular. 
	\item Assume $f$ is monotone and let $S\subseteq T \subseteq I$. Then $R\cup S \subseteq R\cup T$ and therefore,
	$$g(S) =f(R\cup S )\leq f(R\cup T) = g(T).$$
	Thus $g$ is monotone.
	\item If $f$ is modular then both $f$ and $-f$ are submodular. Thus, by the first property both $g$ and $-g$ are submodular and therefore $g$ is modular. 
\end{enumerate}
\end{proof}
\section{The Mulitilinear  Extension of a Modular Function}
\label{sec:multilinear_of_linear}
In this section we show several well known properties of the 
multilinear extension of a linear function.

\begin{lemma}
\label{lem:multilinear_of_linear}
Let $f:2^I \rightarrow \mathbb{R}$ such that $f(S)= a + \sum_{i\in S} \bar{p}_i$ where $a\in \mathbb{R}$ and $\bar{p}\in \mathbb{R}^I$. 
Then, $F$ the  multilinear extension of $f$  satisfies
\begin{equation*}
\forall \bx \in [0,1]^I:~~F(\bx)=a+ \bx\cdot \bar{p}
\end{equation*}
\end{lemma}
\begin{proof}
	We use the notation $$
	\forall S\subseteq I, ~i\in I:~~~~\one_{i\in S} =\begin{cases}1 & i\in S \\ 0 &i\not\in S\end{cases}.$$
	Let $ \bx \in [0,1]^I$, then 
\begin{equation*}
\begin{aligned}
F(\bx)&= \E_{S\sim \bx} \left[ f(S)\right]\\
&= \E_{S\sim \bx} \left[ a + \sum_{i\in S} \bar{p}_i \right]\\
&=a+\E_{S\sim \bx} \left[ \sum_{i\in I}  \one_{i\in S} \cdot  \bar{p}_i \right]\\
&=a+\sum_{i\in I}  \bar{p}_i \cdot \E_{S\sim \bx} \left[ \one_{i\in S}  \right]\\
&= a+\sum_{i\in I} \bx_i \cdot \bar{p}_i= a+ \bx\cdot \bar{p}
\end{aligned}
\end{equation*}
\end{proof}

\cout{
\section{Chernoff Bound}

In the analysis of the algorithm we use the following  Chernoff-like bounds.
\begin{lemma}[Theorem 3.1 in \cite{RKPS06}]
	\label{lem:chernoff}
	Let $X= \sum_{i=1}^{n} X_i \cdot \lambda_i$ where $(X_i)_{i=1}^{n}$ is a sequence of independent Bernoulli random variable and $\lambda_i \in [0,1]$ for $1\leq i \leq n$. Then
	for any $\eps \in (0,1)$ and $\eta\geq \E[X]$ it holds that
	$$\Pr\left(X > (1+\eps) \eta \right)< \exp\left(- \frac{\eps^2}{3} \eta\right)$$
\end{lemma}
}





\end{document}